\documentclass[11pt]{article}
\usepackage[hmargin={2.8cm,2.8cm},vmargin={2.8cm,2.8cm}]{geometry}

\usepackage{amsmath,amssymb}
\usepackage{amsfonts}
\usepackage{amsthm}
\usepackage{mathrsfs}
\usepackage{dsfont}

\usepackage{epic}
\usepackage{eepic}
\usepackage{graphicx}

\usepackage{pgf}
\usepackage{tikz}
\usetikzlibrary{arrows.meta}
\usetikzlibrary{decorations.markings}

\usepackage{mathtools}

\usepackage{bm}

\usepackage{empheq}

\usepackage{stmaryrd}

\usepackage[skip=2pt,font=small]{caption}
\usepackage{hyperref}
\hypersetup{
colorlinks,
linkcolor={blue!60!black},
citecolor={blue!90!black},
urlcolor={blue!90!black}}

\usepackage{footnote}

\usepackage{cite}
\usepackage{enumerate}
\usepackage[curve]{xy}

\usepackage{titlesec}

\titleformat{\paragraph}
{\normalfont\normalsize\bfseries}{\theparagraph}{1em}{}
\titlespacing*{\paragraph}
{0pt}{3.25ex plus 1ex minus .2ex}{1.5ex plus .2ex}

\usepackage[titletoc,title]{appendix}

\newtheorem{theorem}{Theorem}[section]
\newtheorem{lemma}{Lemma}
\newtheorem{remark}{Remark}

\newtheorem{prop}[theorem]{Proposition}

\tikzset{->-/.style={decoration={
  markings,
  mark=at position .68 with {\arrow{Latex}}},postaction={decorate}}}
\tikzset{->>-/.style={decoration={
  markings,
  mark=at position .69 with {\arrow{Latex[sep=-10pt]Latex}}},postaction={decorate}}}
\tikzset{-<-/.style={decoration={
  markings,
  mark=at position .56 with {\arrow{Latex[reversed]}}},postaction={decorate}}}
\tikzset{-<<-/.style={decoration={
  markings,
  mark=at position .57 with {\arrow{Latex[reversed,sep=-10pt]Latex[reversed]}}},postaction={decorate}}}
  
\tikzset{-|-/.style={decoration={
  markings,
  mark=at position .51 with {\arrow{Bar}}},postaction={decorate}}}
\tikzset{-||-/.style={decoration={
  markings,
  mark=at position .49 with {\arrow{Bar[sep=-5pt] Bar}}},postaction={decorate}}}
\tikzset{-!-/.style={decoration={
  markings,
  mark=at position .51 with {\arrow{Tee Barb[length=4pt]}}},postaction={decorate}}}
\tikzset{-!!-/.style={decoration={
  markings,
  mark=at position .51 with {\arrow{Tee Barb[sep=1pt,length=4pt] Tee Barb[length=4pt]}}},postaction={decorate}}}

\usetikzlibrary{shapes.misc}

\tikzset{cross/.style={cross out, draw=black, minimum size=2*(#1-\pgflinewidth), inner sep=0pt, outer sep=0pt},
cross/.default={1pt}}

\newcommand{\ds}{\displaystyle}

\renewcommand{\author}[1]{\large\rm #1\\ \bigskip}
\newcommand{\address}[1]{{\normalsize\it #1\\}\bigskip}
\renewcommand{\title}[1]{\bigskip\bigskip\Large\bf #1\bigskip\bigskip\\}

\newcommand{\Bigpsi}[3]{\phantom{\Psi}_2 \kern -.05em
\Psi_2\left(\genfrac{}{}{0pt}{}{#1}{#2}\biggl|#3\right)}

\newcommand{\bea}{\begin{eqnarray}}
\newcommand{\eea}{\end{eqnarray}}

\newcommand{\ii}{\mathsf{i}}

\newcommand{\oW}{\overline{W}}

\renewcommand{\r}{{\mathsf r}}
\newcommand{\lag}{{\mathcal L}}
\newcommand{\ol}{\overline{\lag}}

\newcommand{\olamh}{\hat{\olamh}}

\newcommand{\q}{{\mathsf q}}
\newcommand{\p}{{\mathsf p}}

\newcommand{\iW}{{W}}

\newcommand{\bW}{{V}}

\newcommand{\bb}{\mathsf{b}}

\newcommand{\bp}{{\mathbf \p}}
\newcommand{\bq}{{\mathbf \q}}
\newcommand{\br}{{\mathbf \r}}
\newcommand{\bu}{{\mathbf u}}
\newcommand{\bv}{{\mathbf v}}
\newcommand{\bbw}{{\mathbf w}}

\newcommand{\bw}{{\mathbf w}}
\newcommand{\bx}{{\mathbf x}}
\newcommand{\cbx}{{x}}
\newcommand{\by}{{\mathbf y}}
\newcommand{\cby}{{y}}

\def\EXP{\textrm{{\large e}}}

\def\re{\mathop{\hbox{\rm Re}}\nolimits}
\def\im{\mathop{\hbox{\rm Im}}\nolimits}

\newcommand{\x}{{\boldsymbol{x}}}
\newcommand{\y}{{\boldsymbol{y}}}
\newcommand{\al}{{\bm{\alpha}}}
\newcommand{\bt}{{\bm{\beta}}}
\newcommand{\gm}{{\bm{\gamma}}}
\newcommand{\tht}{{\bm{\theta}}}
\newcommand{\ph}{{\bm{\phi}}}

\newcommand{\ow}{\overline{W}}

\renewcommand{\r}{{\mathsf r}}

\newcommand{\Log}{{\textrm{Log}\hspace{1pt}}}

\newcommand{\dilog}{{\textrm{Li}_2}}

\newcommand{\lie}{{\textrm{Li}_2}}

\newcommand{\spn}{\bm{\xi}}
\newcommand{\cspn}{\xi}

\def\EXP{\textrm{{\large e}}}

\def\re{\mathop{\hbox{\rm Re}}\nolimits}
\def\im{\mathop{\hbox{\rm Im}}\nolimits}

\newcounter{app}
\newcounter{sapp}[app]

\begin{document}

\vglue 2cm

\begin{center}

\title{Quasi-classical expansion of a hyperbolic solution to the star-star relation and multicomponent 5-point difference equations}
\author{Andrew P.~Kels}
\address{School of Mathematics and Statistics, The University of New South Wales, Sydney, NSW 2052, Australia}

\end{center}

\begin{abstract}

The quasi-classical expansion of a multicomponent spin solution of the star-star relation with hyperbolic Boltzmann weights is investigated.  
The equations obtained in a quasi-classical limit provide $n-1$-component extensions of certain scalar 5-point equations (corresponding to $n=2$) that were previously investigated by the author in the context of integrability and consistency of equations on face-centered cubics.  

\end{abstract}

\section{Introduction}

Edge-interaction models of statistical mechanics involve interactions between pairs of spin variables that are connected by edges of a lattice.  A prominent example of such a model is the Ising model.  A Yang-Baxter equation for such models takes a special form known as the star-triangle relation \cite{Baxter:1982zz,McCoyBook,PerkSTR,PerkYBEs,Bazhanov:2016ajm}.  If the Boltzmann weights of a model satisfy the star-triangle relation this implies that the transfer matrices of the model commute and in principle one can solve the model using the methods of Baxter \cite{Baxter:1982zz}.  Thus a solution of the star-triangle relation may be used to define an integrable lattice model of statistical mechanics.  Some solutions of the star-triangle relation that generalise the Ising model have been extensively studied including the Fateev-Zamolodchikov model \cite{Fateev:1982wi}, Kashiwara-Miwa model \cite{Kashiwara:1986tu}, and chiral Potts model \cite{AuYang:1987zc,Baxter:1987eq}.  Recently, several generalisations of these solutions of the star-triangle relation have been obtained which have deep and interesting connections with hypergeometric integrals, supersymmetric quantum field theories, discrete integrable systems, and other related areas \cite{Bazhanov:2007mh,Bazhanov:2010kz,Spiridonov:2010em,Zabrodin:2010qm,Derkachov:2012iv,Chicherin:2012yn,Yamazaki:2013nra,Kels:2015bda,Kashaev:2015nya,Gahramanov:2016wxi,GahramanovKels,Kels:2018xge,Yamazaki:2018xbx,Sarkissian:2018ppc,Spiridonov:2019uuw,Derkachov:2019ynh,Derkachov:2019tzo,Eren:2019ibl,de-la-Cruz-Moreno:2020xop,Bazhanov:2022wdj,SchlosserSTR}.  

Besides the star-triangle relation, there is another relation that implies integrability for edge-interaction models of statistical mechanics known as the star-star relation \cite{Baxter:1986df,Bazhanov:1992jqa,KashaevStarSquare,Baxter:1997tn}.  In this case, one may take products of Boltzmann weights to reformulate the model as either an interaction-round-a-face (IRF) model or a vertex model, and the star-star relation implies that a form of the Yang-Baxter equation is satisfied in either case.  Such a vertex formulation was notably used by Bazhanov and Stroganov \cite{Bazhanov:1989nc} to derive an $R$-matrix of the chiral Potts model as the intertwiner of two $L$-operators associated with the $R$-matrix of the six-vertex model, and this idea was subsequently generalised to construct the $sl(n)$ chiral Potts model from an $n$-state model associated with the $\mathcal{U}_{\q}(sl(n))$ algebra \cite{Bazhanov:1990qk}.  As for the star-triangle relation, general new solutions of the star-star relation have recently been obtained having interesting connections with hypergeometric integrals, supersymmetric quantum field theories, and discrete integrable systems \cite{Bazhanov:2011mz,Bazhanov:2013bh,Yamazaki:2013nra,Kels:2017toi,Yamazaki:2018xbx,Catak:2021coz,Mullahasanoglu:2023nes}.

An important connection has been developed \cite{Bazhanov:2007mh,Bazhanov:2010kz,Bazhanov:2011mz,Bazhanov:2016ajm,Kels:2018xge,Kels:2020zjn} between the above star-triangle and star-star relations for integrable edge-interaction models of statistical mechanics and different types of integrable systems which can be classed as integrable partial difference equations \cite{hietarinta_joshi_nijhoff_2016}. The latter provide lattice analogues of integrable partial differential equations, an important example of which is the Korteweg-de Vries (KdV) equation.  The characteristics of integrability for these discrete soliton equations are quite different to that for edge-interaction models, and include the existence of Lax pairs, B\"acklund transformations, and measuring a low `entropy' of the evolution.  Thus, understanding the aforementioned connection can be expected to provide insight into how the characteristics of the two different types integrable systems are related and potentially be a step towards a unified definition of integrability of these systems.

Specifically, the key to connecting these two different types of integrable systems lies in the quasi-classical expansion.  For example, through this expansion the leading asymptotics of the star-triangle are described by additive three-leg forms associated to 4-point discrete soliton equations \cite{ABS,ABS2}, and the latter equations are then equivalent to the exponential of the saddle-point equation of the star-triangle relation \cite{Bazhanov:2016ajm,Kels:2018xge}.  Furthermore, the quasi-classical asymptotics of the partition function of the edge-interaction model are described by discrete Laplace-type equations \cite{AdlerPlanarGraphs,BobSurQuadGraphs,ABS,MR2467378} associated to the discrete soliton equations.  The saddle-point equations in this case correspond to systems of 5-point difference equations that evolve in the square lattice \cite{Kels:2020zjn}.  

Thus, if a solution to the star-triangle or star-star relation is known, the quasi-classical limit can provide a way to derive and investigate difference equations, where the latter should be expected to inherit characteristics associated with integrability of the lattice models.   For example, in the quasi-classical limit Yang-Baxter equations implied by the star-star relation may be reinterpreted as consistency relations satisfied by 5-point difference equations \cite{Kels:2020zjn}.  In turn, it is possible to reinterpret these consistency conditions in terms of Lax pairs \cite{KelsLax1,KelsLax2}, similarly to Lax pairs that arise from consistency of 4-point difference equations \cite{nijhoffwalker,BobSurQuadGraphs}.  Thus this shows how consistency and Lax pairs for integrable difference equations are connected to forms of the Yang-Baxter equation for integrable lattice models of statistical mechanics.

While the quasi-classical limit for models satisfying the star-triangle relation and the connection to integrable difference equations has been relatively well developed \cite{Bazhanov:2007mh,Bazhanov:2010kz,Bazhanov:2011mz,Bazhanov:2016ajm,Kels:2018xge}, much less has been done for the star-star relation.  A couple of results for this area include the quasi-classical limit of a general elliptic solution of the star-star relation investigated by Bazhanov and Sergeev \cite{Bazhanov:2011mz}, as well as the study of a degeneration of hyperbolic analogues of equations associated to the latter star-star relation \cite{Kels:2018qzx}.  The purpose of this paper is to study the latter hyperbolic equations themselves in more detail.  Namely, this paper will investigate the quasi-classical expansion for a multicomponent spin solution of the star-star relation corresponding to an identity for hyperbolic hypergeometric integrals associated to the $A_n$ root system \cite{Rains2009} and the resulting system of $n-1$-component difference equations.  The latter equations will be found to provide multicomponent extensions of known integrable scalar 5-point equations that were previously studied in the context of the IRF formulation of edge-interaction models and consistency on face-centered cubics \cite{Kels:2020zjn}. 

The layout of this paper is as follows.  In Section \ref{sec:YBE} the edge-interaction model of statistical mechanics associated to the star-star relation is introduced.  In Section \ref{sec:hyperbolic}, the particular hyperbolic solution of the star-star relation is given, and the quasi-classical limit is investigated from which the $n-1$ component difference equations are derived.  The rational analogues of the multicomponent difference equations are also given.  In Section \ref{sec:latticeconsistency}, the quasi-classical expansion of the partition function is considered and consistency for the multicomponent difference equations is formulated in terms of the consistency-around-a-face-centered-cube property.

\section{Edge-interaction model and star-star relation}\label{sec:YBE}

\subsection{Edge interaction model on checkerboard square lattice}
\label{sec:msmodeldef}

It is convenient to define the model of statistical mechanics on the checkerboard square lattice \cite{Baxter:1986df} which is pictured in Figure \ref{fig-lattice}.  The vertices of the lattice may be decomposed into black and white subsets, where the set of black vertices will be denoted $V^{(B)}$ and the set of white vertices will be denoted $V^{(W)}$.  Associated to the checkerboard square lattice is the medial rapidity lattice consisting of horizontal and vertical directed rapidity lines.   Rapidity variables $p$ and $p'$ are assigned alternately to horizontally directed rapidity lines, and rapidity variables $q$ and $q'$ are assigned alternately to vertically directed rapidity lines, as indicated in Figure \ref{fig-lattice}.

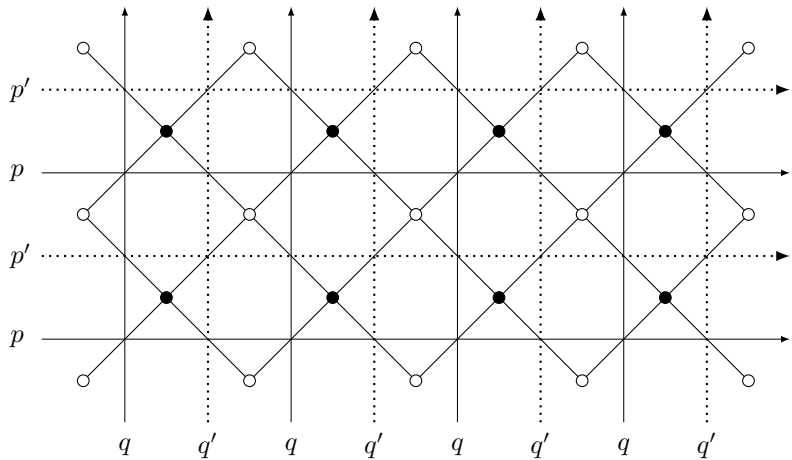
\begin{figure}[htb]
\centering
\begin{tikzpicture}[scale=1.1]

\draw[-] (-0.5,-0.5)--(3.5,3.5);
\draw[-] (-0.5,3.5)--(3.5,-0.5);
\draw[-] (-0.5,1.5)--(1.5,3.5)--(3.5,1.5)--(1.5,-0.5)--(-0.5,1.5);
\draw[-] (-4.5,-0.5)--(-0.5,3.5);
\draw[-] (-4.5,3.5)--(-0.5,-0.5);
\draw[-] (-4.5,1.5)--(-2.5,3.5)--(-0.5,1.5)--(-2.5,-0.5)--(-4.5,1.5);
\foreach \x in {-4,-2,...,2}{
\draw[-latex,very thin] (\x,-1) -- (\x,4);
\fill[white!] (\x,-1) circle (0.08pt)
node[below=2.4pt]{\color{black}\small$q$};}
\foreach \x in {-3,-1,...,3}{
\draw[-latex,thick,dotted] (\x,-1) -- (\x,4);
\fill[white!] (\x,-1) circle (0.08pt)
node[below=0.05pt]{\color{black}\small$q'$};}
\foreach \y in {1,3}{
\draw[-latex,thick,dotted] (-5,\y) -- (4,\y);
\fill[white!] (-5,\y) circle (0.08pt)
node[left=0.05pt]{\color{black}\small$p'$};}
\foreach \y in {0,2}{
\draw[-latex,very thin] (-5,\y) -- (4,\y);
\fill[white!] (-5,\y) circle (0.08pt)
node[left=2.9pt]{\color{black}\small$p$};}
\foreach \y in {-0.5,1.5,3.5}{
\foreach \x in {-4.5,-2.5,...,3.5}{
\filldraw[fill=white,draw=black] (\x,\y) circle (2.0pt);}}
\foreach \y in {0.5,2.5}{
\foreach \x in {-3.5,-1.5,...,2.5}{
\filldraw[fill=black,draw=black] (\x,\y) circle (2.0pt);}}


\end{tikzpicture}

\caption{Checkerboard square lattice and directed rapidity lines.}
\label{fig-lattice}
\end{figure}

To each vertex $i\in V^{(B)}\cup V^{(W)}$ is assigned an $n$-component spin variable
\begin{equation}\label{spinvals}
\spn_i=\bigl((\cspn_i)_1,\ldots,(\cspn_i)_n\bigr)\in\mathbb{R}^n,
\end{equation}
subject to the following constraint on components
\begin{equation}
\sum_{a=1}^n(\cspn_i)_a=0.
\end{equation}
Thus a spin has $n-1$ independent components and $n=2$ corresponds to the scalar case.

To specify the model, Boltzmann weights will be assigned which characterise the interactions between nearest-neighbour pairs of spins that are connected by edges of the lattice.  As seen in Figure \ref{fig-lattice}, the intersection of two rapidity lines distinguish four different types of edges of the checkerboard square lattice which are shown in Figure \ref{fig-crosses}.  The corresponding sets of four different types of edges in the checkerboard square lattice will be respectively denoted $E^{(i)}$, $i=1,2,3,4$, as indicated in Figure \ref{fig-crosses}.

\begin{figure}[htb]
\centering
\begin{tikzpicture}[scale=2.2]
\draw[-,thick] (-0.5,2)--(0.5,2);
\draw[->,thick,dotted] (0.4,1.6)--(-0.4,2.4);
\fill[white!] (0.4,1.6) circle (0.01pt)
node[below=0.5pt]{\color{black}\small$q'$};
\draw[->] (-0.4,1.6)--(0.4,2.4);
\fill[white!] (-0.4,1.6) circle (0.01pt)
node[below=2.4pt]{\color{black}\small$p$};
\filldraw[fill=black,draw=black] (-0.5,2) circle (0.8pt)
node[left=3pt]{\color{black}\small $i$};
\filldraw[fill=white,draw=black] (0.5,2) circle (0.8pt)
node[right=3pt]{\color{black}\small $j$};

\fill (0,1.1) circle(0.01pt)
node[below=0.05pt]{\color{black} $E^{(1)}$};

\begin{scope}[xshift=55pt]
\draw[-,thick] (-0.5,2)--(0.5,2);
\draw[->] (0.4,1.6)--(-0.4,2.4);
\fill[white!] (0.4,1.6) circle (0.01pt)
node[below=3.3pt]{\color{black}\small$q$};
\draw[->,thick,dotted] (-0.4,1.6)--(0.4,2.4);
\fill[white!] (-0.4,1.6) circle (0.01pt)
node[below=0.5pt]{\color{black}\small$p'$};
\filldraw[fill=white,draw=black] (-0.5,2) circle (0.8pt)
node[left=3pt]{\color{black}\small $j$};
\filldraw[fill=black,draw=black] (0.5,2) circle (0.8pt)
node[right=3pt]{\color{black}\small $i$};

\fill (0,1.1) circle(0.01pt)
node[below=0.05pt]{\color{black} $E^{(2)}$};
\end{scope}

\begin{scope}[xshift=105pt,yshift=57pt]
\draw[-,thick] (0,-0.5)--(0,0.5);
\draw[->] (-0.4,-0.4)--(0.4,0.4);
\fill[white!] (-0.4,-0.4) circle (0.01pt)
node[below=3.2pt]{\color{black}\small$p$};
\draw[->] (0.4,-0.4)--(-0.4,0.4);
\fill[white!] (0.4,-0.4) circle (0.01pt)
node[below=3.2pt]{\color{black}\small$q$};
\filldraw[fill=white,draw=black] (0,-0.5) circle (0.8pt)
node[below=4pt]{\color{black}\small $j$};
\filldraw[fill=black,draw=black] (0,0.5) circle (0.8pt)
node[above=4pt]{\color{black}\small $i$};

\fill (0,-0.9) circle(0.01pt)
node[below=0.05pt]{\color{black} $E^{(3)}$};
\end{scope}

\begin{scope}[xshift=150pt,yshift=57pt]
\draw[-,thick] (0,-0.5)--(0,0.5);
\draw[->,thick,dotted] (-0.4,-0.4)--(0.4,0.4);
\fill[white!] (-0.4,-0.4) circle (0.01pt)
node[below=0.5pt]{\color{black}\small$p'$};
\draw[->,thick,dotted] (0.4,-0.4)--(-0.4,0.4);
\fill[white!] (0.4,-0.4) circle (0.01pt)
node[below=0.5pt]{\color{black}\small$q'$};
\filldraw[fill=black,draw=black] (0,-0.5) circle (0.8pt)
node[below=4pt]{\color{black}\small $i$};
\filldraw[fill=white,draw=black] (0,0.5) circle (0.8pt)
node[above=4pt]{\color{black}\small $j$};

\fill (0,-0.9) circle(0.01pt)
node[below=0.05pt]{\color{black} $E^{(4)}$};
\end{scope}

\end{tikzpicture}
\caption{Four different types of edges belonging to the sets $E^{(1)}$, $E^{(2)}$, $E^{(3)}$, and $E^{(4)}$, respectively.}
\label{fig-crosses}
\end{figure}
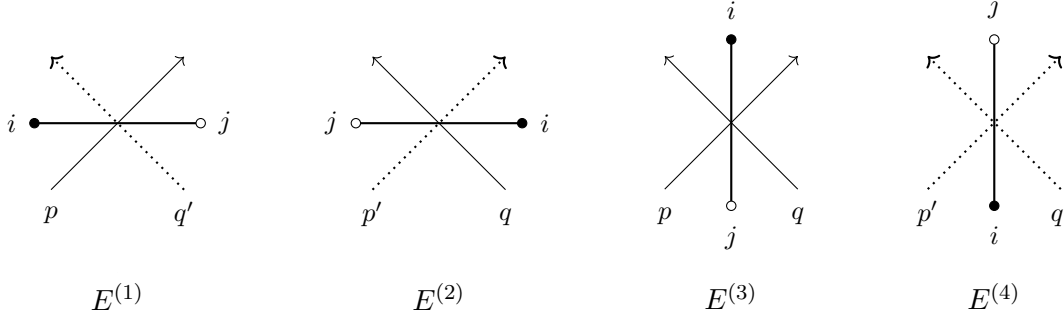

Four different expressions for the Boltzmann weights will be assigned to these different edges.  Namely, an edge $(ij)\in E^{(k)}$ connecting two vertices $i\in V^{(B)}$ and $j\in V^{(W)}$ is assigned a Boltzmann weight
\begin{equation}\label{BW-crosses}
W_{p-q'}(\spn_i,\spn_j),\quad W_{p'-q}(\spn_i,\spn_j), \quad
\ow_{p-q}(\spn_j,\spn_i),\quad \ow_{p'-q'}(\spn_j,\spn_i),
\end{equation}
for $k=1,2,3,4$, respectively.  Finally, another Boltzmann weight denoted
\begin{equation}
S(\spn_i),
\end{equation} 
is assigned to each vertex $i\in V^{(B)}\cup V^{(W)}$.  To be considered a physical model of statistical mechanics the Boltzmann weights should take real and positive values, however, for the purposes of this paper this is not needed.  Note that it is typically the case (at least for integrable cases) that the Boltzmann weights satisfy
\begin{equation}\label{BWinv}
W_\alpha(\spn_i,\spn_j)W_{-\alpha}(\spn_j,\spn_i)=1,
\end{equation}
and are related by
\begin{equation}\label{BWeta}
\ow_{\alpha}(\spn_i,\spn_j)=W_{\eta-\alpha}(\spn_i,\spn_j),
\end{equation}
for some parameter $\eta$ known as the crossing parameter.

Let $V^{(int)}$ denote the set of vertices interior to the lattice.  For example, in the lattice of Figure \ref{fig-lattice} the interior vertices each have four nearest neighbours and the remaining vertices are on the boundary.  Then the partition function of the model $Z$ may formally be written as
\begin{equation}
\label{Zdef}
\begin{split}
Z = \int_{\mathbb{R}^{n-1}}\cdots\int_{\mathbb{R}^{n-1}}\prod_{k\in V^{(int)}} d\spn_k \prod_{i\in V^{(B)}\cup V^{(W)}}S(\spn_i) 
\prod_{(ij)\in E^{(1)}} W_{p-q'}(\spn_i,\spn_j)
\prod_{(ij)\in E^{(2)}} W_{p'-q}(\spn_i,\spn_j)\phantom{,} \\
\times\prod_{(ij)\in E^{(3)}} \ow_{p-q}(\spn_j,\spn_i)
\prod_{(ij)\in E^{(4)}} \ow_{p'-q'}(\spn_j,\spn_i),
\end{split}
\end{equation}
where $d\spn_k$ denotes $\prod_{a=1}^{n-1}d(\spn_k)_a$ and the boundary spins are kept fixed.

The above is a rather general construction of a multicomponent spin edge-interaction lattice model of statistical mechanics and closely follows the construction used by Bazhanov and Sergeev for their solution of the star-star relation given in terms of the elliptic gamma function \cite{Bazhanov:2011mz}. Specific expressions for Boltzmann weights used in this paper will be presented in the next section, which give the hyperbolic analogue of \cite{Bazhanov:2011mz}.  
It remains to introduce the star-star relation as a condition of integrability for the above model, and this will be considered in the remainder of this section.  

\subsection{Star-star relation}
Each interior white (black) vertex in the square lattice of Figure \ref{fig-lattice} is connected by four edges to four different black (white) vertices.  The corresponding two four-edge `star' configurations of vertices and edges is shown in Figure \ref{fig-IRF}.  These star diagrams are assigned the following IRF Boltzmann weights according to Figure \ref{fig-crosses} and \eqref{BW-crosses}
\begin{align}
\label{V1}
\ds \bW_{\bp\bq}^{(B)}(\spn_a,\spn_b,\spn_c,\spn_d)=\ds\int_{\mathbb{R}^{n-1}} d\spn_i S(\spn_i)\,\overline{\iW}_{p-q}(\spn_c,\spn_i)\,\overline{\iW}_{p'-q'}(\spn_b,\spn_i)\,\iW_{p'-q}(\spn_i,\spn_a)\,\iW_{p-q'}(\spn_i,\spn_d), \\
\label{V2}
\ds \bW_{\bp\bq}^{(W)}(\spn_a,\spn_b,\spn_c,\spn_d)=\int_{\mathbb{R}^{n-1}} d\spn_i S(\spn_i)\, \overline{\iW}_{p-q}(\spn_i,\spn_b)\,  \overline{\iW}_{p'-q'}(\spn_i,\spn_c)\,\iW_{p'-q}(\spn_d,\spn_i)\,\iW_{p-q'}(\spn_a,\spn_i),
\end{align}
for the diagram on the left and right of Figure \ref{fig-IRF} respectively.

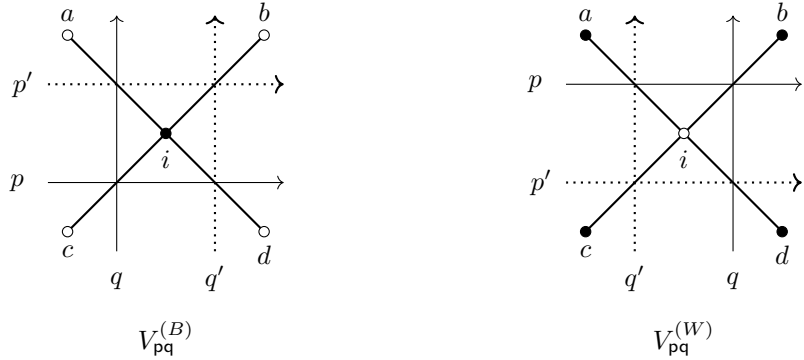
\begin{figure}[htb]
\centering
\begin{tikzpicture}[scale=1.3]

\draw[-,thick] (-1,-1)--(1,1);
\draw[-,thick] (1,-1)--(-1,1);

\filldraw[fill=white!,draw=black!] (-1,-1) circle (1.5pt)
node[below=1.5pt]{\color{black}\small$c$};
\filldraw[fill=white!,draw=black!] (1,-1) circle (1.5pt)
node[below=1.5pt]{\color{black}\small$d$};
\filldraw[fill=white!,draw=black!] (-1,1) circle (1.5pt)
node[above=1.5pt]{\color{black}\small$a$};
\filldraw[fill=white!,draw=black!] (1,1) circle (1.5pt)
node[above=1.5pt]{\color{black}\small$b$};

\filldraw[fill=black!,draw=black!] (0,0) circle (1.5pt)
node[below=2.5pt]{\color{black}\small$i$};

\draw[->,thick,dotted] (-1.2,0.5) -- (1.2,0.5);
\draw[white!] (-1.2,0.5) circle (0.01pt)
node[left=1.5pt]{\color{black}\small$p'$};
\draw[->] (-1.2,-0.5) -- (1.2,-0.5);
\draw[white!] (-1.3,-0.5) circle (0.01pt)
node[left=1.5pt]{\color{black}\small$p$};
\draw[->] (-0.5,-1.2) -- (-0.5,1.2);
\draw[white!] (-0.5,-1.27) circle (0.01pt)
node[below=1.5pt]{\color{black}\small$q$};
\draw[->,thick,dotted] (0.5,-1.2) -- (0.5,1.2);
\draw[white!] (0.5,-1.2) circle (0.01pt)
node[below=1.5pt]{\color{black}\small$q'$};

\draw[white!] (0,-1.8) circle (0.01pt)
node[below=0.1pt]{\color{black}\small$\bW^{(B)}_{\bp\bq}$};

\begin{scope}[xshift=150pt]

\draw[-,thick] (-1,-1)--(1,1);
\draw[-,thick] (1,-1)--(-1,1);

\filldraw[fill=black!,draw=black!] (-1,-1) circle (1.5pt)
node[below=1.5pt]{\color{black}\small$c$};
\filldraw[fill=black!,draw=black!] (1,-1) circle (1.5pt)
node[below=1.5pt]{\color{black}\small$d$};
\filldraw[fill=black!,draw=black!] (-1,1) circle (1.5pt)
node[above=1.5pt]{\color{black}\small$a$};
\filldraw[fill=black!,draw=black!] (1,1) circle (1.5pt)
node[above=1.5pt]{\color{black}\small$b$};

\filldraw[fill=white!,draw=black!] (0,0) circle (1.5pt)
node[below=2.5pt]{\color{black}\small$i$};

\draw[->] (-1.2,0.5) -- (1.2,0.5);
\draw[white!] (-1.3,0.5) circle (0.01pt)
node[left=1.5pt]{\color{black}\small$p$};
\draw[->,thick,dotted] (-1.2,-0.5) -- (1.2,-0.5);
\draw[white!] (-1.2,-0.5) circle (0.01pt)
node[left=1.5pt]{\color{black}\small$p'$};
\draw[->,thick,dotted] (-0.5,-1.2) -- (-0.5,1.2);
\draw[white!] (-0.5,-1.2) circle (0.01pt)
node[below=1.5pt]{\color{black}\small$q'$};
\draw[->] (0.5,-1.2) -- (0.5,1.2);
\draw[white!] (0.5,-1.27) circle (0.01pt)
node[below=1.5pt]{\color{black}\small$q$};

\draw[white!] (0,-1.8) circle (0.01pt)
node[below=0.1pt]{\color{black}\small$\bW^{(W)}_{\bp\bq}$};

\end{scope}
\end{tikzpicture}

\caption{The two types of four-edge stars that appear in the checkerboard square lattice.}
\label{fig-IRF}
\end{figure}



A sufficient condition of integrability for the edge-interaction model on the checkerboard lattice is that the Boltzmann weights satisfy the star-star relation, which implies that the transfer matrices of the model commute \cite{Baxter:1997tn}.  In terms of the two types of IRF Boltzmann weights \eqref{V1} and \eqref{V2}, the expression for the star-star relation for this model may be written as
\begin{align}
\label{ssrdef}
\iW_{q'-q}(\spn_d,\spn_c)\,\iW_{q-q'}(\spn_a,\spn_b) \,\bW_{\bp\bq}^{(B)}(\spn_a,\spn_b,\spn_c,\spn_d)\,=\,\iW_{p-p'}(\spn_b,\spn_d)\,\iW_{p'-p}(\spn_c,\spn_a)\,\bW_{\bp\bq}^{(W)}(\spn_a,\spn_b,\spn_c,\spn_d).
\end{align}
This relation may also be regarded as a duality transformation between the two configurations of edges shown in Figure \ref{fig-IRF}, up to some prefactors. These prefactors are products of edge Boltzmann weights that do not have corresponding edges that appear in the checkerboard square lattice.  In fact, the edge Boltzmann weights that appear on the left hand side would be associated to edges that connect two black vertices, and those on the right hand side would be associated to edges that connect two white vertices. Taking into account only the crossing of rapidities shown in Figure \ref{fig-crosses} and corresponding assignment of Boltzmann weights in \eqref{BW-crosses}, the above expression for the star-star relation is an equality for Boltzmann weights assigned to the configurations of edges and rapidity lines shown in Figure \ref{fig-ssr}.

\begin{figure}[htb]
\centering
\begin{tikzpicture}[scale=1.3]

\draw[-,thick] (-1,-1)--(1,1);
\draw[-,thick] (1,-1)--(-1,1);
\draw[-,thick] (-1,1)--(1,1);
\draw[-,thick] (-1,-1)--(1,-1);

\filldraw[fill=white!,draw=black!] (-1,-1) circle (1.5pt)
node[below=1.5pt]{\color{black}\small$c$};
\filldraw[fill=white!,draw=black!] (1,-1) circle (1.5pt)
node[below=1.5pt]{\color{black}\small$d$};
\filldraw[fill=white!,draw=black!] (-1,1) circle (1.5pt)
node[above=1.5pt]{\color{black}\small$a$};
\filldraw[fill=white!,draw=black!] (1,1) circle (1.5pt)
node[above=1.5pt]{\color{black}\small$b$};

\filldraw[fill=black!,draw=black!] (0,0) circle (1.5pt)
node[below=2.5pt]{\color{black}\small$i$};

\draw[->,thick,dotted] (-1.7,0.5) -- (1.7,0.5);
\draw[white!] (-1.7,0.5) circle (0.01pt)
node[left=1.5pt]{\color{black}\small$p'$};
\draw[->] (-1.7,-0.5) -- (1.7,-0.5);
\draw[white!] (-1.8,-0.5) circle (0.01pt)
node[left=1.5pt]{\color{black}\small$p$};
\draw[->,thick,dotted] (-0.5,-1.7) .. controls (-0.5,-1.6) and (-0.4,-1.2) .. (0,-1) .. controls (0.2,-0.9) and (0.5,-0.6) .. (0.5,-0.5) .. controls (0.7,-0.4) and (0.7,0.4) .. (0.5,0.5) .. controls (0.5,0.6) and (0.2,0.9) .. (0,1) .. controls (-0.4,1.2) and (-0.5,1.6) .. (-0.5,1.7);
\draw[white!] (-0.5,-1.7) circle (0.01pt)
node[below=1.5pt]{\color{black}\small$q'$};
\draw[->] (0.5,-1.7) .. controls (0.5,-1.6) and (0.4,-1.2) .. (0,-1) .. controls (-0.2,-0.9) and (-0.5,-0.6) .. (-0.5,-0.5) .. controls (-0.7,-0.4) and (-0.7,0.4) .. (-0.5,0.5) .. controls (-0.5,0.6) and (-0.2,0.9) .. (0,1) .. controls (0.4,1.2) and (0.5,1.6) ..  (0.5,1.7);
\draw[white!] (0.5,-1.77) circle (0.01pt)
node[below=1.5pt]{\color{black}\small$q$};


\draw[white!] (2.2,0) circle (0.01pt)
node[right=0.1pt]{\color{black}=};

\begin{scope}[xshift=150pt]

\draw[-,thick] (-1,-1)--(1,1);
\draw[-,thick] (1,-1)--(-1,1);
\draw[-,thick] (-1,-1)--(-1,1);
\draw[-,thick] (1,-1)--(1,1);

\filldraw[fill=black!,draw=black!] (-1,-1) circle (1.5pt)
node[below=1.5pt]{\color{black}\small$c$};
\filldraw[fill=black!,draw=black!] (1,-1) circle (1.5pt)
node[below=1.5pt]{\color{black}\small$d$};
\filldraw[fill=black!,draw=black!] (-1,1) circle (1.5pt)
node[above=1.5pt]{\color{black}\small$a$};
\filldraw[fill=black!,draw=black!] (1,1) circle (1.5pt)
node[above=1.5pt]{\color{black}\small$b$};

\filldraw[fill=white!,draw=black!] (0,0) circle (1.5pt)
node[below=2.5pt]{\color{black}\small$j$};

\draw[->,thick,dotted] (-1.7,0.5) .. controls (-1.6,0.5) and (-1.2,0.4) .. (-1,0) .. controls (-0.9,-0.2) and (-0.6,-0.5) .. (-0.5,-0.5) .. controls (-0.4,-0.7) and (0.4,-0.7) .. (0.5,-0.5) .. controls (0.6,-0.5) and (0.9,-0.2) .. (1,0) .. controls (1.2,0.4) and (1.6,0.5) .. (1.7,0.5);
\draw[white!] (-1.7,0.5) circle (0.01pt)
node[left=1.5pt]{\color{black}\small$p'$};
\draw[->] (-1.7,-0.5) .. controls (-1.6,-0.5) and (-1.2,-0.4) .. (-1,0) .. controls (-0.9,0.2) and (-0.6,0.5) .. (-0.5,0.5) .. controls (-0.4,0.7) and (0.4,0.7) .. (0.5,0.5) .. controls (0.6,0.5) and (0.9,0.2) .. (1,0) .. controls (1.2,-0.4) and (1.6,-0.5) .. (1.7,-0.5);
\draw[white!] (-1.8,-0.5) circle (0.01pt)
node[left=1.5pt]{\color{black}\small$p$};
\draw[->,thick,dotted] (-0.5,-1.7) -- (-0.5,1.7);
\draw[white!] (-0.5,-1.7) circle (0.01pt)
node[below=1.5pt]{\color{black}\small$q'$};
\draw[->] (0.5,-1.7) -- (0.5,1.7);
\draw[white!] (0.5,-1.77) circle (0.01pt)
node[below=1.5pt]{\color{black}\small$q$};


\end{scope}
\end{tikzpicture}

\caption{Star-star relation}
\label{fig-ssr}
\end{figure}
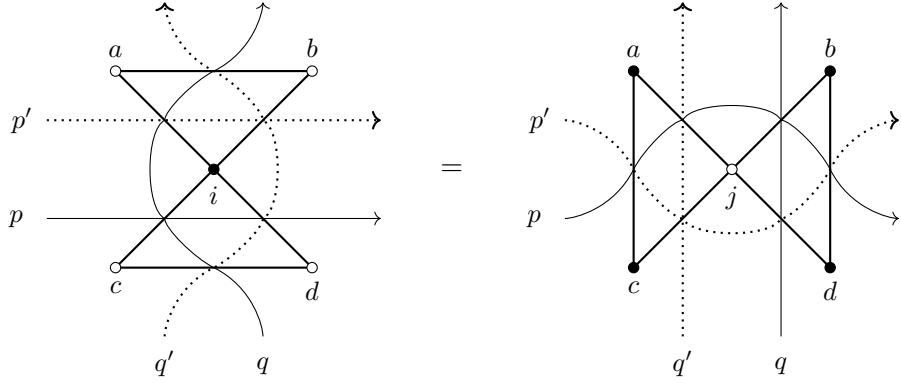

The star-star relation is an equation that involves six $n$-component spin variables, where the four `boundary' spins $\spn_a,\spn_b,\spn_c,\spn_d$ are kept fixed, and the two `interior' spins (one on the left hand side and one on the right hand side) are integrated over their components.  Thus the star-star relation may also be regarded simply as the equality of partition functions for simple models of statistical mechanics associated to the two configurations shown on the left and right hand sides of Figure \ref{fig-ssr}.

Finally, consider the following renormalised IRF Boltzmann weight
\begin{equation}
\label{weight-IRF}
\begin{split}
\bW_{\bp\bq}(\spn_a,\spn_b,\spn_c,\spn_d)=\left(\frac{\ds\iW_{q'-q}(\spn_d,\spn_c)\,\iW_{p'-p}(\spn_d,\spn_b)}{\ds\iW_{q'-q}(\spn_b,\spn_a)\,  \iW_{p'-p}(\spn_c,\spn_a)}\right)^{\frac{1}{2}}\, \bW_{\bp\bq}^{(1)}(\spn_a,\spn_b,\spn_c,\spn_d)\phantom{,} \\
=\left(\frac{\ds\iW_{q'-q}(\spn_b,\spn_a) \, \iW_{p'-p}(\spn_c,\spn_a)}{\ds\iW_{q'-q}(\spn_d,\spn_c) \,\iW_{p'-p}(\spn_d,\spn_b)}\right)^{\frac{1}{2}}\ \bW_{\bp\bq}^{(2)}(\spn_a,\spn_b,\spn_c,\spn_d),
\end{split}
\end{equation}
where the second equality follows as a consequence of the star-star relation \eqref{ssrdef} and the inversion relation \eqref{BWinv}.
The star-star relation \eqref{ssrdef} implies the following form of the Yang-Baxter equation for IRF weights \cite{Baxter:1997tn,Bazhanov:2011mz}
\begin{equation}
\label{YBE-IRF}
\begin{split}
\ds \int_{\mathbb{R}^{n-1}} d\spn_i\, S(\spn_i)\;\bW_{\bp\bq}(\spn_c,\spn_i,\spn_e,\spn_d)\,\bW_{\bp\br}(\spn_i,\spn_b,\spn_d,\spn_f)\,\bW_{\bq\br}(\spn_c,\spn_g,\spn_i,\spn_b)\phantom{,} \\ 
=\ds \int_{\mathbb{R}^{n-1}} d\spn_j\, S(\spn_j)\;\bW_{\bq\br}(\spn_e,\spn_j,\spn_d,\spn_f)\,\bW_{\bp\br}(\spn_c,\spn_g,\spn_e,\spn_j)\,\bW_{\bp\bq}(\spn_g,\spn_b,\spn_j,\spn_f).
\end{split}
\end{equation}
This is an equation for 14 spin $n$-component spin variables, where the six boundary spins $\spn_b,\spn_c,\spn_d,\spn_e,\spn_f,\spn_g$ are fixed, and the eight interior spins (four on the left hand side and four on the right hand side) are integrated over their components.

%


\section{Multicomponent lattice model with hyperbolic Boltzmann weights}\label{sec:hyperbolic}


In this section, specific Boltzmann weights will be given in terms of the hyperbolic gamma function/non-quantum compact dilogarithm, which is defined here by \cite{Faddeev:1994fw,Ruijsenaars:1997:FOA}
\begin{equation}\label{HGF}
    \Gamma_h(z;\bb)=\exp\left(\int_0^\infty\frac{dx}{x}\Bigl(\frac{\ii z}{x}-\frac{\sinh(2\ii zx)}{2\sinh(x\bb)\sinh(x\bb^{-1})}\Bigr)\right),\quad |\im(z)|<\re(\eta_h),
\end{equation}
where $\eta_h$ will be referred to as the crossing parameter, defined by
\begin{equation}
\label{hypeta}
\eta_h=\frac{\bb+\bb^{-1}}{2}.
\end{equation}
For the purposes here, the parameter $\bb$ will take positive values
\begin{equation}
\label{bbdef}
\bb>0,
\end{equation}
such that the crossing parameter \eqref{hypeta} is also positive.  
%
%
The function \eqref{HGF} has appeared previously in several different but related forms including the multiple sine functions \cite{Barnes:1901,Barnes1904,shintani1976,kurokawa1991} 
the non-compact quantum dilogarithm \cite{Faddeev:1994fw,FaddeevKashaev}, and the hyperbolic gamma function \cite{Ruijsenaars:1997:FOA}.

The hyperbolic gamma function \eqref{HGF} satisfies the inversion relation
\begin{equation}
\label{HGFidents}
\Gamma_h(z;\bb)=\frac{1}{\Gamma_h(-z;\bb)},
\end{equation}
and difference equations
\begin{equation}
\label{HGFdif}
\displaystyle
\frac{\Gamma_h(z-\ii\bb;\bb)}{\Gamma_h(z;\bb)}=2\cosh(\pi(2z-\ii\bb^{-1})/(2\bb)),\quad\frac{\Gamma_h(z-\ii\bb^{-1};\bb)}{\Gamma_h(z;\bb)}=2\cosh(\pi(2z-\ii\bb)\bb/2),
\end{equation}
The latter relations may be used to extend \eqref{HGF} as a meromorphic function on $z\in\mathbb{C}$ \cite{Ruijsenaars:1997:FOA}.

\subsection{Star-star relation}

Let $\spn_i=\bigl((\cspn_i)_1,\ldots,(\cspn_i)_n\bigr)$ denote an $n$-component spin variable, and $p_a,q_b$, $a,b=1,2$, denote rapidity parameters, which respectively take values
\begin{equation}
\label{hypvals}
\spn_i\in\mathbb{R}^n,\quad\sum_{i=0}^n(\cspn_i)_j=0,\qquad 0<p_a-q_b<\eta_h,\quad a,b=1,2.
\end{equation}

The Boltzmann weights are defined in terms of the hyperbolic gamma function by
\begin{equation}
\begin{split}
S(\spn_i)=&\prod_{1\leq a<b\leq n}\Gamma_h(-\ii\eta_h+(\cspn_i)_a-(\cspn_i)_b)\Gamma_h(-\ii\eta_h-(\cspn_i)_a+(\cspn_i)_b) \\
=&
\prod_{1\leq a<b\leq n}4\sinh\Bigl(\pi\bigl((\cspn_i)_a-(\cspn_i)_b\bigr)\bb^{-1}\Bigr)
\sinh\Bigl(\pi\bigl((\cspn_i)_a-(\cspn_i)_b\bigr)\bb\Bigr),
\end{split}
\end{equation}
and
\begin{equation}
    W_{p-q}(\spn_i,\spn_j)=\prod_{a,b=1}^n\Gamma_h((\cspn_i)_a-(\cspn_i)_b+\ii(p-q)), \qquad
    \oW_{p-q}(\spn_i,\spn_j)=W_{\eta_h-(p-q)}(\spn_i,\spn_j).
\end{equation}

By definition, the above Boltzmann weights are each invariant under permutations of the components of the spin variables $\spn_i$, $\spn_j$.  The Boltzmann weights also satisfies the reflection symmetry
\begin{equation}
    W_{p-q}(\spn_i,\spn_j)W_{q-p}(\spn_j,\spn_i)=1.
\end{equation}

Finally, define the following IRF Boltzmann weights
\begin{equation}\label{IRFBWhyp}
\begin{split}
    W^{(B)}(\spn)=&
    \int_{\mathbb{R}^{n-1}}\!\!d\spn_fS(\spn_f)W_{p_2-q_1}(\spn_f,\spn_i)\oW_{p_2-q_2}(\spn_j,\spn_f)\oW_{p_1-q_1}(\spn_k,\spn_f)W_{p_1-q_2}(\spn_f,\spn_l), \\[0.1cm]
    W^{(W)}(\spn)=&
    \int_{\mathbb{R}^{n-1}}\!\!d\spn_gS(\spn_g)W_{p_1-q_2}(\spn_i,\spn_g)\oW_{p_1-q_1}(\spn_g,\spn_j)\oW_{p_2-q_2}(\spn_g,\spn_k)W_{p_2-q_1}(\spn_l,\spn_g),
\end{split}
\end{equation}
where $d\spn_f$ denotes $d(\cspn_f)_1\ldots d(\cspn_f)_{n-1}$.
\begin{theorem}\label{thm:ssr}
The following star-star relation holds
\begin{equation}
\label{hypssr}
\begin{split}
W_{p_1-p_2}(\spn_i,\spn_k)W_{q_1-q_2}(\spn_i,\spn_j)W^{(B)}(\spn)
=W_{p_1-p_2}(\spn_j,\spn_l)W_{q_1-q_2}(\spn_k,\spn_l)W^{(W)}(\spn).
\end{split}
\end{equation}
\end{theorem}

The univariate $n=2$ case of this star-star relation has been investigated in connection with supersymmetric gauge theories \cite{Catak:2021coz}, while for general $n\geq2$ the above star-star relation may be obtained from the hyperbolic limit \cite{Rains2009} of the elliptic solution of the star-star relation that was first obtained by Bazhanov and Sergeev \cite{Bazhanov:2011mz} and proven in \cite{Bazhanov:2013bh} by connecting it to identities of elliptic hypergometric integrals associated to the $A_n$ root system \cite{RainsT}. 

\subsection{Quasi-classical limit}\label{sec:QCL}

The quasi-classical limit of the univariate ($n=2$) case of \eqref{hypssr} has previously been shown to lead to 5-point difference equations which are consistent on a face-centered cubic unit cell \cite{Kels:2020zjn}.  Here mulitcomponent 5-point difference equations shall be derived from the quasi-classical limit of \eqref{hypssr} for general $n$.

First, define the following parameter
\begin{equation}
\label{hbar}
\hbar=2\pi \bb^2\;.
\end{equation}
A quasi-classical expansion may be taken by scaling the variables and parameters as \cite{Bazhanov:2007mh}
\begin{equation}
\label{qclas}
    \spn_i\to\frac{\bx_i}{\sqrt{2\pi\hbar}},\qquad p_j\to\frac{u_j}{\sqrt{2\pi\hbar}},\qquad q_j\to\frac{v_j}{\sqrt{2\pi\hbar}},
\end{equation}
and considering $\hbar\to0$.  The classical variables $\bx_i$, and parameters $u_j$, $v_j$, $j=1,2$, appearing above take values
\begin{equation}
\label{hypcvals}
\bx_i\in\mathbb{R}^n,\qquad 0<u_j-v_k<\pi,\quad j,k=1,2,
\end{equation}
and the sum-to-zero condition on variables \eqref{hypvals} will be fixed as
\begin{equation}
    \bx_i=\bigl((\cbx_i)_1,\ldots,(\cbx_i)_{n-1},-X\bigr),
\end{equation}
where
\begin{equation}\label{sumtozero}
    X=\sum_{a=1}^{n-1}(\cbx_i)_a,
\end{equation}
so that the independent components of $\bx_i$, are $(\cbx_i)_a$, $a=1,\ldots,n-1$.  Note that the condition on parameters in \eqref{hypcvals} that is inherited from \eqref{hypvals} is quite restrictive, but can effectively be dropped once the desired difference equations have been obtained from the quasi-classical limit. 

In the quasi-classical limit \eqref{qclas} the leading asymptotics of \eqref{HGF} are given by \cite{FaddeevKashaev}
\begin{equation}
\label{HGFqcl}
\Log\Gamma_h(z(2\pi\bb)^{-1};\bb)=-\ii\hbar^{-1}\Bigl(\dilog(-\EXP^{z})+\frac{\pi^2}{12}-\frac{z^2}{4}\Bigr)+O(\hbar),\qquad \im(z)<\pi,
\end{equation}
where $\lie(z)$ is the dilogarithm function, defined for $\mathbb{C}\setminus[1,\infty)$ by
\begin{equation}
\dilog(z)=-\int^z_0dx\,\frac{\Log(1-x)}{x}.
\end{equation}
Using \eqref{HGFqcl}, the leading asymptotics of Boltzmann weights may be written in the form
\begin{equation}\begin{split}
    \log S(\spn_i)=&-\ii\hbar^{-1} C(\bx_i)+O(1), \\
        \log W_{p-q}(\spn_i,\spn_j)=&-\ii\hbar^{-1}\lag_{u-v}(\bx_i,\bx_j)+O(1), \\
    \log \oW_{p-q}(\spn_i,\spn_j)=&-\ii\hbar^{-1}\ol_{u-v}(\bx_i,\bx_j),+O(1),
\end{split}\end{equation}
where
\begin{equation}\label{lagdefs}
\begin{split}
 C(\bx_i)=&-\ii\pi\sum_{1\leq a<b\leq n}\bigl((\cbx_i)_a-(\cbx_i)_b\bigr), \\
     \lag_{u-v}(\bx_i,\bx_j)=&\frac{(\pi^2-3(u-v)^2)n^2}{12}+\frac{n}{4}\sum_{a=1}^n\bigl(((\cbx_i)_a)^2+((\cbx_j)_a)^2\bigr)
    +\sum_{a,b=1}^n\lie\bigl(-\EXP^{(\cbx_i)_a-(\cbx_j)_b+\ii(u-v)}\bigr), \\
    \ol_{u-v}(\bx_i,\bx_j)=&\lag_{\pi-u+v}(\bx_i,\bx_j).
\end{split}
 \end{equation}

The Lagrangian function satisfies an anti-symmetry relation 
\begin{equation}
\label{clasinversion}
\lag_\alpha(\bx_i,\bx_j) =-\lag_{-\alpha}(\bx_j,\bx_i),
\end{equation}
which is the classical analogue of \eqref{BWinv}.

Due to the appearance of exponentials it will be convenient to implement a change of variables
\begin{equation}\label{hypcovdef}
    (\cby_i)_a=\EXP^{(\cbx_i)_a},\quad \alpha_j=\EXP^{\ii u_j},\quad \beta_j=\EXP^{\ii v_j},
\end{equation}
such that the multicomponent variables $\by_i$ are given by
\begin{equation}\label{yvardef}
    \by_i=\bigl((\cby_i)_1,\ldots,(\cby_i)_{n-1},Y_i^{-1}\bigr),
\end{equation}
and
\begin{equation}\label{prodtozero}
    Y_i=\prod_{a=1}^{n-1}(\cby_i)_a.
\end{equation}
The point is that exponentials of the partial derivatives of the above Lagrangian functions with respect to the components $(\cspn_i)_a$ may be expressed as rational functions of components of the variables $\by_i$.  These rational functions are expressed in terms of
\begin{equation}\label{phidefhyp}
    \phi_a(\by_i,\by_j;\alpha,\beta)=\Biggl(\frac{(y_i)_a}{Y_i}\Biggr)^{\frac{n}{2}}
        \prod_{b=1}^n\frac{\alpha-\beta Y_i(\cby_j)_b}{\alpha(\cby_i)_a-\beta(\cby_j)_b},
\end{equation}
as follows
\begin{prop}\label{prop:phicov}
The exponentials of partial derivatives of $\lag$ and $\ol$ defined in \eqref{lagdefs} may be written in terms of $\phi_a$ as
\begin{equation}
    \begin{array}{lll}
        &\ds\exp\left(\frac{\partial \lag_{u-v}(\bx_i,\bx_j)}{\partial (\cbx_i)_a}\right)=\phi_a(\by_i,\by_j;\alpha,-\beta), \quad
        &\ds\exp\left(\frac{\partial \ol_{u-v}(\bx_i,\bx_j)}{\partial (\cbx_i)_a}\right)=\phi_a(\by_i,\by_j;\beta,\alpha), \\[0.3cm]
        &\ds\exp\left(\frac{\partial \lag_{u-v}(\bx_j,\bx_i)}{\partial (\cbx_i)_a}\right)=\phi_a(\by_i,\by_j;\beta,-\alpha)^{-1}, \quad
        &\ds\exp\left(\frac{\partial \ol_{u-v}(\bx_j,\bx_i)}{\partial (\cbx_i)_a}\right)=\phi_a(\by_i,\by_j;\alpha,\beta)^{-1}.
    \end{array}
\end{equation}
\end{prop}
\begin{proof}
This may be shown by direct computation mainly using the fact that the dilogarithm satisfies
\begin{equation}
\exp\left(\frac{\partial}{\partial x}\dilog(\EXP^{x})\right)=1-\EXP^{x}.
\end{equation}
\end{proof}

\subsection{ Systems of multicomponent 5-point difference equations}
The above quasi-classical expansions will be applied to the products of Boltzmann weights that appear in the star-star relation of Theorem \ref{thm:ssr} to derive the desired 5-point difference equations.

In the following, it is convenient to shift two parameters as
\begin{equation}\label{varshft}
v_1\to v_1+\pi,\qquad v_2\to v_2+\pi,
\end{equation}
which will make the expressions obtained below uniform in the parameters.

From the quasi-classical expansions of the Boltzmann weights obtained above, the quasi-classical expansion of the star-star relation \eqref{hypssr} has the form
\begin{equation}\label{ssrqcl}
\begin{split}
\int_{\mathbb{R}^{n-1}}\frac{d\bx_f}{(2\pi\hbar)^{\frac{n-1}{2}}}\,\EXP^{-\frac{\ii}{\hbar}\left(\lag_{u_1-u_2}(\bx_i,\bx_k)+\lag_{v_1-v_2}(\bx_i,\bx_j)+\lag^{(B)}_{\bu\bv}(\bx_f;\bx_i,\bx_j,\bx_k,\bx_l)\right)+O(1)}\phantom{,} \\
=
\int_{\mathbb{R}^{n-1}}\frac{d\bx_g}{(2\pi\hbar)^{\frac{n-1}{2}}}\,\EXP^{-\frac{\ii}{\hbar}\left(\lag_{u_1-u_2}(\bx_j,\bx_l)+\lag_{v_1-v_2}(\bx_k,\bx_l)+\lag^{(W)}_{\bu\bv}(\bx_g;\bx_i,\bx_j,\bx_k,\bx_l)\right)+O(1)},
\end{split}
\end{equation}
where $d\bx_f$ denotes $\prod_{a=1}^{n-1}(\cbx_f)_a$, and
\begin{equation}
\begin{split}
 \lag^{(B)}_{\bu\bv}(\bx_f;\bx_i,\bx_j,\bx_k,\bx_l)=
 C(\bx_f)+\lag_{u_2-v_1}(\bx_f,\bx_i)+\lag_{u_1-v_2}(\bx_f,\bx_l)\phantom{,} \\
 +\ol_{u_2-v_2}(\bx_j,\bx_f)+\ol_{u_1-v_1}(\bx_k,\bx_f), \\
 \lag^{(W)}_{\bu\bv}(\bx_g;\bx_i,\bx_j,\bx_k,\bx_l)=
 C(\bx_g)+\lag_{u_1-v_2}(\bx_i,\bx_g)+\lag_{u_2-v_1}(\bx_l,\bx_g)\phantom{,} \\
 +\ol_{u_1-v_1}(\bx_g,\bx_j)+\ol_{u_2-v_2}(\bx_g,\bx_k).
\end{split}
 \end{equation}
The saddle-point equations on the left and right hand sides of \eqref{ssrqcl} are respectively given by
\begin{equation}\label{sadpt}
\frac{\partial}{\partial (\cbx_f)_a} \lag^{(B)}_{\bu\bv}(\bx_f;\bx_i,\bx_j,\bx_k,\bx_l)=0,\qquad
\frac{\partial}{\partial (\cbx_g)_a} \lag^{(W)}_{\bu\bv}(\bx_g;\bx_i,\bx_j,\bx_k,\bx_l)=0,\qquad  a=1,\ldots,n-1.
\end{equation}
Let $\al=(\alpha_1,\alpha_2)$ and $\bt=(\beta_1,\beta_2)$, and define the following function in terms of \eqref{phidefhyp}
\begin{equation}\label{Adefhyp}
    A_a(\by_f;\by_i,\by_j,\by_k,\by_l;\al,\bt)=
    \frac{\phi_a(\by_f,\by_i;\alpha_2,\beta_1)\phi_a(\by_f,\by_l;\alpha_1,\beta_2)}
    {\phi_a(\by_f,\by_j;\alpha_2,\beta_2)\phi_a(\by_f,\by_k;\alpha_1,\beta_1)}.
\end{equation}

Define also
\begin{equation}
\hat{\al}=(\alpha_2,\alpha_1),\quad \hat{\bt}=(\beta_2,\beta_1).
\end{equation}
The function $A_a(\by_f;\by_i,\by_j,\by_k,\by_l;\al,\bt)$ satisfies the following obvious symmetries
\begin{equation}\label{Asyms}
\begin{split}
A_a(\by_f;\by_i,\by_j,\by_k,\by_l;\al,\bt)
&=A_a(\by_f;\by_l,\by_k,\by_j,\by_i;\hat{\al},\hat{\bt}) \\
&=A_a(\by_f;\by_j,\by_i,\by_l,\by_k;\al,\hat{\bt})^{-1} \\
&=A_a(\by_f;\by_k,\by_l,\by_i,\by_j;\hat{\al},\bt)^{-1}.
\end{split}
\end{equation}

If $A_a(\by_f;\by_i,\by_j,\by_k,\by_l;\al,\bt)$ is associated to the vertices of Figure \ref{fig-phidef}, the above correspond to symmetries under rotations and reflections of the square with dashed edges.

\begin{figure}[htb]
\centering
\begin{tikzpicture}[scale=1.3]

\draw[dashed] (-1,-1)--(-1,1)--(1,1)--(1,-1)--(-1,-1);
\draw[-,thick] (-1,-1)--(1,1);
\draw[-,thick] (1,-1)--(-1,1);

\filldraw[fill=black!,draw=black!] (-1,-1) circle (1.5pt)
node[below=1.5pt]{\color{black}\small$k$};
\filldraw[fill=black!,draw=black!] (1,-1) circle (1.5pt)
node[below=1.5pt]{\color{black}\small$l$};
\filldraw[fill=black!,draw=black!] (-1,1) circle (1.5pt)
node[above=1.5pt]{\color{black}\small$i$};
\filldraw[fill=black!,draw=black!] (1,1) circle (1.5pt)
node[above=1.5pt]{\color{black}\small$j$};

\filldraw[fill=black!,draw=black!] (0,0) circle (1.5pt)
node[below=2.5pt]{\color{black}\small$h$};



\end{tikzpicture}

\caption{A vertex configuration for \eqref{Adefhyp}.}
\label{fig-phidef}
\end{figure}
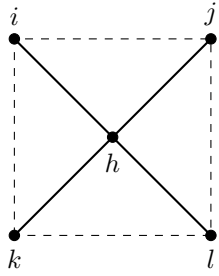

Using Proposition \ref{prop:phicov} and the variables \eqref{yvardef} , one may write the saddle-point equations \eqref{sadpt} as follows.
\begin{prop}\label{prop:IRFphicov}
The exponentials of the partial derivatives of the saddle-point equations \eqref{sadpt} may be written in terms of $A_a$ as
\begin{align}\label{sadpt2}
    \exp\left(\frac{\partial}{\partial (\cbx_f)_a} \lag^{(B)}_{\bu\bv}(\bx_f;\bx_i,\bx_j,\bx_k,\bx_l)\right) 
    =A_a(\by_f;\by_i,\by_j,\by_k,\by_l;\al,\bt),\phantom{^{-1}} \\
    \exp\left(\frac{\partial}{\partial (\cbx_g)_a} \lag^{(W)}_{\bu\bv}(\bx_g;\bx_i,\bx_j,\bx_k,\bx_l)\right) 
    =A_a(\by_g;\by_i,\by_k,\by_j,\by_l;\bt,\al)^{-1},
\end{align}
for $ a=1,\ldots,n-1$.
\end{prop}
The expressions appearing in Proposition \ref{prop:IRFphicov} are rational functions that are each multilinear in the components of the four variables $\by_i,\by_j,\by_k,\by_l$, respectively (note that this does not imply a unique or rational solution for $n>2$), and also are invariant under permutations of the components of any of these four variables.

Thus, through Proposition \ref{prop:IRFphicov} the $n-1$ saddle-point equations \eqref{sadpt} become the following equations
\begin{align}\label{fcqhypa}
A_a(\by_f;\by_i,\by_j,\by_k,\by_l;\al,\bt)=1,\qquad a=1,\ldots,n-1, \\ \label{fcqhypa2}
A_a(\by_f;\by_i,\by_k,\by_j,\by_l;\bt,\al)=1,\qquad a=1,\ldots,n-1,
\end{align}
which may be interpreted as $n-1$-component systems of 5-point equations for the variables $\by_f,\by_i,\by_j,\by_k,\by_l.$  

\subsubsection{The case $n=2$}


For $n=2$, the variables \eqref{yvardef} take the form $\by_i=\bigl(y_i,(y_i)^{-1}\bigr)$.  Then the equation \eqref{fcqhypa} is equivalent to
\begin{equation}\label{a311}
A(y_h;y_i,y_j,y_k,y_l;\al,\bt)=
\frac{a(y_h,y_i;\alpha_2,\beta_1)a(y_h,y_l;\alpha_1,\beta_2)}
{a(y_h,y_j;\alpha_2,\beta_2)a(y_h,y_k;\alpha_1,\beta_1)}=1,
\end{equation}
where
\begin{equation}
a(y_i,y_j;\alpha,\beta)=
\frac{\alpha-\beta y_iy_j}{\alpha y_i -\beta y_j}
\frac{\alpha-\beta y_i/y_j}{\alpha y_i -\beta/y_j}.
\end{equation}
In terms of the variables
\begin{equation}
\overline{y}_i=\sqrt{y_i^2-1},
\end{equation}
this becomes
\begin{equation}
a(\overline{y}_i,\overline{y}_j;\alpha,\beta)=
\frac{\alpha^2+\beta^2\overline{y}_i^2-2\alpha\beta\overline{y}_iy_j}
{\beta^2+\alpha^2\overline{y}_i^2-2\alpha\beta\overline{y}_iy_j}.
\end{equation}
The equation \eqref{fcqhypa2} gives a similar expression. Then \eqref{a311} corresponds to the multiplicative four-leg form of an equation labelled $A3_{(1)}$ that was previously derived in the context of consistency of 5-point equations on a face-centered cubic \cite{Kels:2020zjn}.  Thus for $n>2$, the system of equations \eqref{fcqhypa} (or \eqref{fcqhypa2}) provide a multicomponent extension of this equation.

\subsubsection{The case $n=3$}

For $n=3$, the variables \eqref{yvardef} take the form $\by_i=\bigl((\cby_i)_1,(\cby_i)_2,((\cby_i)_1(\cby_i)_2)^{-1}\bigr)$.  In the following it will be useful to use the following hatted variables which exchange the $1$ and $2$ components:
\begin{equation}
\hat{\by}_h=\bigl((y_h)_2,(y_h)_1,((y_h)_1(y_h)_2)^{-1}).
\end{equation}

Define the following three multivariate polynomials
\begin{equation}
\begin{split}
P_0(\by_i,\by_j,\by_k;\by_h;\al,\bt)=&
G_1(\by_i,\by_j,\by_k;\by_h;\al,\bt)\bigl(\tfrac{\alpha_1}{\beta_2}\bigr)^2/(y_h)_3 -
G_2(\by_i,\by_j,\by_k;\by_h;\al,\bt)\bigl(\tfrac{\alpha_1}{\beta_2}(y_h)_2\bigr)^2, \\
P_1(\by_i,\by_j,\by_k;\by_h;\al,\bt)=&
G_1(\by_i,\by_j,\by_k;\by_h;\al,\bt)\Bigl(\bigl((y_h)_3\bigr)^{-3}-\bigl(\tfrac{\alpha_1}{\beta_2}\bigr)^3\Bigr) \\ -&  
G_2(\by_i,\by_j,\by_k;\by_h;\al,\bt)\Bigl(1-\bigl(\tfrac{\alpha_1}{\beta_2}(y_h)_2\bigr)^3\Bigr), \\
P_2(\by_i,\by_j,\by_k;\by_h;\al,\bt)=&
G_2(\by_i,\by_j,\by_k;\by_h;\al,\bt)\tfrac{\alpha_1}{\beta_2}(y_h)_2 -
G_1(\by_i,\by_j,\by_k;\by_h;\al,\bt)\tfrac{\alpha_1}{\beta_2}\bigl((y_h)_3\bigr)^{-2},
\end{split}
\end{equation}
where
\begin{equation}
\begin{split}
G_1(\by_i,\by_j,\by_k;\by_h;\al,\bt)=\prod_{a=1}^3((y_i)_a/(y_h)_3-\tfrac{\alpha_2}{\beta_1})((y_j)_a-(y_h)_2\tfrac{\alpha_2}{\beta_2})((y_k)_a-(y_h)_2\tfrac{\alpha_1}{\beta_1}), \\
G_2(\by_i,\by_j,\by_k;\by_h;\al,\bt)=\prod_{a=1}^3((y_i)_a-(y_h)_2\tfrac{\alpha_2}{\beta_1})((y_j)_a/(y_h)_3-\tfrac{\alpha_2}{\beta_2})((y_k)_a/(y_h)_3-\tfrac{\alpha_1}{\beta_1}).
\end{split}
\end{equation}

The equations \eqref{fcqhypa} for $n=3$ will be expressed in terms of the following polynomial in variables $r$ and $s$
\begin{equation}\label{multpolydef}
\begin{split}
P(\by_i,\by_j,\by_k;\by_h;\al,\bt;r,s) 
&= P_0(\by_i,\by_j,\by_k;\by_h;\al,\bt)(r^2+s) \\
&+P_1(\by_i,\by_j,\by_k;\by_h;\al,\bt)r \\
&+P_2(\by_i,\by_j,\by_k;\by_h;\al,\bt)(rs+1).
\end{split}
\end{equation}

\begin{lemma}\label{lem:hyp}
For $n=3$, the pair of equations \eqref{fcqhypa} may be written in the equivalent form
\begin{equation}\label{n3eqs}
\begin{split}
P(\by_i,\by_j,\by_k;\by_h;\al,\bt;r_l,s_l)=0, \\
P(\by_i,\by_j,\by_k;\hat{\by}_h;\al,\bt;r_l,s_l)=0,
\end{split}
\end{equation}
where $r_l=(y_l)_1(y_l)_2$ and $s_l=(y_l)_1+(y_l)_2$.
\end{lemma}

\begin{proof}
This may be obtained from \eqref{fcqhypa} by writing the equations as polynomials and noting that $y_l$ and $z_l$ only appear in the symmetric combinations $y_l+z_l$ or $y_lz_l$.  
\end{proof}

\begin{remark}
One may use the symmetries \eqref{Asyms} to obtain equivalent equations written in terms of the pairs $(r_i,s_i)$, $(r_j,s_j)$, or $(r_k,s_k)$. 
\end{remark}

Define $F(\by_i,\by_j,\by_k;\by_h;\al,\bt;x)$ to be the following cubic in the variable $x$
\begin{equation}\label{cubichypdef}
\begin{split}
F(\by_i,\by_j,\by_k;\by_h;\al,\bt;x) &=
F_0(\by_i,\by_j,\by_k;\by_h;\al,\bt) \\
&+F_1(\by_i,\by_j,\by_k;\by_h;\al,\bt)x \\
&+F_2(\by_i,\by_j,\by_k;\by_h;\al,\bt)x^2 \\
&-F_0(\by_i,\by_j,\by_k;\by_h;\al,\bt)x^3,
\end{split}
\end{equation}
where
\begin{equation}
\begin{split}
F_0(\by_i,\by_j,\by_k;\by_h;\al,\bt)&=P_2(\by_i,\by_j,\by_k;\by_h;\al,\bt)P_0(\by_i,\by_j,\by_k;\hat{\by}_h;\al,\bt) \\ &- P_0(\by_i,\by_j,\by_k;\by_h;\al,\bt)P_2(\by_i,\by_j,\by_k;\hat{\by}_h;\al,\bt), \\
F_1(\by_i,\by_j,\by_k;\by_h;\al,\bt)&=P_0(\by_i,\by_j,\by_k;\by_h;\al,\bt)P_1(\by_i,\by_j,\by_k;\hat{\by}_h;\al,\bt) \\ &- P_1(\by_i,\by_j,\by_k;\by_h;\al,\bt)P_0(\by_i,\by_j,\by_k;\hat{\by}_h;\al,\bt), \\
F_2(\by_i,\by_j,\by_k;\by_h;\al,\bt)&=P_2(\by_i,\by_j,\by_k;\by_h;\al,\bt)P_1(\by_i,\by_j,\by_k;\hat{\by}_h;\al,\bt) \\ &- P_1(\by_i,\by_j,\by_k;\by_h;\al,\bt)P_2(\by_i,\by_j,\by_k;\hat{\by}_h;\al,\bt).
\end{split}
\end{equation}

\begin{prop}\label{prop:hyproot}
Let $t_a$, $a=1,2,3$ denote the three roots of the cubic equation \eqref{cubichypdef}.  The solutions $\bigl((y_l)_1,(y_l)_2\bigr)$ of the 5-point equation \eqref{n3eqs} are given by the six pairs of distinct roots
\begin{equation}
\bigl((y_l)_1,(y_l)_2\bigr)=(t_a,t_b),\qquad a,b=1,2,3,\quad a\neq b.
\end{equation}
\end{prop}

\begin{proof}



After eliminating $s_l$ from the equations \eqref{n3eqs} it is seen that $r_l$ must satisfy the following cubic equation
\begin{equation}
q_0+q_1r_l+q_2r_l^2-q_0r_l^3=0,
\end{equation}
where
\begin{equation}
\begin{split}
q_0&=P_2(\by_i,\by_j,\by_k;\by_h)P_0(\by_i,\by_j,\by_k;\hat{\by}_h) - P_0(\by_i,\by_j,\by_k;\by_h)P_2(\by_i,\by_j,\by_k;\hat{\by}_h), \\
q_1&=P_1(\by_i,\by_j,\by_k;\by_h)P_2(\by_i,\by_j,\by_k;\hat{\by}_h) - P_2(\by_i,\by_j,\by_k;\by_h)P_1(\by_i,\by_j,\by_k;\hat{\by}_h), \\
q_2&=P_1(\by_i,\by_j,\by_k;\by_h)P_0(\by_i,\by_j,\by_k;\hat{\by}_h) - P_0(\by_i,\by_j,\by_k;\by_h)P_1(\by_i,\by_j,\by_k;\hat{\by}_h).
\end{split}
\end{equation}
Equivalently, $1/r_l$ must satisfy
\begin{equation}
q_0+\hat{q}_1r_l^{-1}+\hat{q}_2r_l^{-2}-q_0r_l^{-3}=0,
\end{equation}
where
\begin{equation}
\hat{q}_1=-q_2,\qquad \hat{q}_2=-q_1.
\end{equation}

Since the equations \eqref{fcqhypa} are symmetric under permutations of $(y_l)_1$, $(y_l)_2$, and $1/r_l=((y_l)_1(y_l)_2)^{-1}$, both $(y_l)_1$ and $(y_l)_2$ individually must also be solutions of the cubic equation
\begin{equation}\label{solpoly}
q_0+\hat{q}_1x+\hat{q}_2x^2-q_0x^3=0.
\end{equation}
Furthermore, since the product of roots of this equation are equal to one, $(y_l)_1$ and $(y_l)_2$ clearly must be distinct.  Thus by symmetry the solutions of the equations \eqref{n3eqs} are the six pairs
\begin{equation}\label{hypsols}
((y_l)_1,(y_l)_2)=(t_a,t_b),\qquad a,b=1,2,3,\quad a\neq b,
\end{equation}
where $t_a$, $a=1,2,3$, are the three roots of the cubic equation \eqref{solpoly}.

\end{proof}

\begin{remark}
While the solutions are not rational, the solutions of \eqref{n3eqs} are unique up to permutations of the components of the variables, as expected from the permutation symmetry of the equations.
\end{remark}

\subsubsection{Rational limit}
Introduce a parameter $\epsilon$ and define the variables and parameters 
\begin{equation}
\bbw_h=(\EXP^{(y_h)_1\epsilon},\EXP^{(y_h)_2\epsilon}),\qquad \tht=(\EXP^{\alpha_1\epsilon},\EXP^{\alpha_2\epsilon}),\; \ph=(\EXP^{\beta_1\epsilon},\EXP^{\beta_2\epsilon}).
\end{equation}

A rational limit may be taken by substituting exponentials for the components of the variables and parameters for the system of equations \eqref{fcqhypa} in the form
\begin{equation}\label{fcqhypaepsilon}
A_a\bigl(\bbw_h;\bbw_i,\bbw_j,\bbw_k,\bbw_l;\tht,\ph\bigr)=1,\qquad a=1,\ldots,n-1,
\end{equation}
and taking the limit $\epsilon\to0$.  

In the following, the variables \eqref{yvardef} will be redefined as
\begin{equation}
\by_i=\bigl((y_i)_1,\ldots,(y_i)_{n-1},-Y_i),
\end{equation}
where
\begin{equation}
Y_i=\sum_{a=1}^{n-1}(y_i)_a.
\end{equation}
Then, in terms of the function
\begin{equation}\label{phidefrat}
    \phi^{(r)}_a(\by_i,\by_j;\alpha,\beta)=
        \prod_{b=1}^n\frac{ Y_i +(\cby_j)_b -\alpha+\beta}{(\cby_i)_a-(\cby_j)_b+\alpha-\beta},
\end{equation}
the limit $\epsilon\to0$ of \eqref{fcqhypaepsilon} is
\begin{equation}\label{fcqrata}
    A^{(r)}_a(\by_h;\by_i,\by_j,\by_k,\by_l;\al,\bt)=1,\qquad a=1,\ldots,n-1,
\end{equation}
where
\begin{equation}
A^{(r)}_a(\by_h;\by_i,\by_j,\by_k,\by_l;\al,\bt)=
    \frac{\phi^{(r)}_a(\by_h,\by_i;\alpha_2,\beta_1)\phi^{(r)}_a(\by_h,\by_l;\alpha_1,\beta_2)}
    {\phi^{(r)}_a(\by_h,\by_j;\alpha_2,\beta_2)\phi^{(r)}_a(\by_h,\by_k;\alpha_1,\beta_1)}.
\end{equation}
 The $n-1$ equations \eqref{fcqrata} 
define another $n-1$-component system of 5-point equations for the variables $\by_h,\by_i,\by_j,\by_k,\by_l.$  Presumably the above rational system will arise from the quasi-classical expansion of a star-star relation related to identities for rational hypergeometric integrals associated to the $A_n$ root system \cite{Rains2009}.

\subsubsection{Rational limit: $n=2$ and $n=3$ cases}
First, for the $n=2$ case the variables are $\by_i=(y_i,-y_i)$ and the function \eqref{phidefrat} is
\begin{equation}
\phi^{(r)}_a(\by_i,\by_j;\alpha,\beta)=\frac{y_i+y_j-\alpha+\beta}{y_i+y_j+\alpha-\beta}\frac{y_i-y_j-\alpha+\beta}{y_i-y_j+\alpha-\beta}.
\end{equation}
In this case the equation \eqref{fcqrata} is equivalent to the multiplicative four-leg form of an equation labelled $A2_{(1;1)}$ from \cite{Kels:2020zjn}.  Thus for $n>2$ the equations \eqref{fcqrata} provide a multicomponent extension of this equation.

For $n=3$, the variables take the form $\by_i=\bigl((y_i)_1,(y_i)_2,-(y_i)_1-(y_i)_2\bigr)$, and the hatted variables will be defined by
\begin{equation}
\hat{\by}_h=\bigl((y_h)_2,(y_h)_1,-(y_h)_1-(y_h)_2\bigr).
\end{equation}
Define the following three multivariate polynomials
\begin{equation}
\begin{split}
P^{(r)}_0(\by_i,\by_j,\by_k;\by_h;\al,\bt)=\,&
G^{(r)}_1(\by_i,\by_j,\by_k;\by_h;\al,\bt)(\beta_2-\alpha_1-(y_h)_3) \\ +\,&
G^{(r)}_2(\by_i,\by_j,\by_k;\by_h;\al,\bt)((y_h)_2+\alpha_1-\beta_2), \\
P^{(r)}_1(\by_i,\by_j,\by_k;\by_h;\al,\bt)=\,&
G^{(r)}_1(\by_i,\by_j,\by_k;\by_h;\al,\bt) -  
G^{(r)}_2(\by_i,\by_j,\by_k;\by_h;\al,\bt), \\
P^{(r)}_2(\by_i,\by_j,\by_k;\by_h;\al,\bt)=\,&
G^{(r)}_1(\by_i,\by_j,\by_k;\by_h;\al,\bt)((y_h)_3+\alpha_1-\beta_2)^3 \\ +\,&
G^{(r)}_2(\by_i,\by_j,\by_k;\by_h;\al,\bt)(\beta_2-\alpha_1-(y_h)_2)^3,
\end{split}
\end{equation}
where
\begin{equation}
\begin{split}
G^{(r)}_1(\by_i,\by_j,\by_k;\by_h;\al,\bt)=\prod_{a=1}^3((y_i)_a-(y_h)_3-\alpha_2+\beta_1)((y_j)_a-(y_h)_2-\alpha_2+\beta_2)\phantom{,} \\[-0.3cm] \times ((y_k)_a-(y_h)_2-\alpha_1+\beta_1), \\
G^{(r)}_2(\by_i,\by_j,\by_k;\by_h;\al,\bt)=\prod_{a=1}^3((y_i)_a-(y_h)_2-\alpha_2+\beta_1)((y_j)_a-(y_h)_3-\alpha_2+\beta_2)\phantom{,} \\[-0.3cm] \times ((y_k)_a-(y_h)_3-\alpha_1+\beta_1).
\end{split}
\end{equation}

The equations \eqref{fcqrata} for $n=3$ will be expressed in terms of the following polynomial in variables $r$ and $s$
\begin{equation}\label{multpolydefrat}
\begin{split}
P^{(r)}(\by_i,\by_j,\by_k;\by_h;\al,\bt;r,s) 
&= P^{(r)}_0(\by_i,\by_j,\by_k;\by_h;\al,\bt)(s^2-r) \\
&+P^{(r)}_1(\by_i,\by_j,\by_k;\by_h;\al,\bt)rs \\
&+P^{(r)}_2(\by_i,\by_j,\by_k;\by_h;\al,\bt).
\end{split}
\end{equation}

The following gives the analogue of Lemma \ref{lem:hyp} and can be shown in a similar way.
\begin{lemma}
For $n=3$, the pair of equations \eqref{fcqrata} may be written in the equivalent form
\begin{equation}\label{n3eqsrat}
\begin{split}
P^{(r)}(\by_i,\by_j,\by_k;\by_h;\al,\bt;r_l,s_l)=0, \\
P^{(r)}(\by_i,\by_j,\by_k;\hat{\by}_h;\al,\bt;r_l,s_l)=0,
\end{split}
\end{equation}
where $r_l=(y_l)_1(y_l)_2$ and $s_l=(y_l)_1+(y_l)_2$.
\end{lemma}





As for the hyperbolic case, solutions of \eqref{n3eqsrat} may be expressed in terms of roots of a cubic equation.  Define $F(\by_i,\by_j,\by_k;\by_h;\al,\bt;x)$ to be the following depressed cubic in the variable $x$
\begin{equation}\label{cubicratdef}
\begin{split}
F^{(r)}(\by_i,\by_j,\by_k;\by_h;\al,\bt;x) &=
F^{(r)}_0(\by_i,\by_j,\by_k;\by_h;\al,\bt) \\
&+F^{(r)}_1(\by_i,\by_j,\by_k;\by_h;\al,\bt)x \\
&+F^{(r)}_3(\by_i,\by_j,\by_k;\by_h;\al,\bt)x^3,
\end{split}
\end{equation}
where
\begin{equation}
\begin{split}
F^{(r)}_0(\by_i,\by_j,\by_k;\by_h;\al,\bt)&=P_2(\by_i,\by_j,\by_k;\by_h;\al,\bt)P_0(\by_i,\by_j,\by_k;\hat{\by}_h;\al,\bt) \\ &- P_0(\by_i,\by_j,\by_k;\by_h;\al,\bt)P_2(\by_i,\by_j,\by_k;\hat{\by}_h;\al,\bt), \\
F^{(r)}_1(\by_i,\by_j,\by_k;\by_h;\al,\bt)&=P_1(\by_i,\by_j,\by_k;\by_h;\al,\bt)P_0(\by_i,\by_j,\by_k;\hat{\by}_h;\al,\bt) \\ &- P_0(\by_i,\by_j,\by_k;\by_h;\al,\bt)P_1(\by_i,\by_j,\by_k;\hat{\by}_h;\al,\bt), \\
F^{(r)}_3(\by_i,\by_j,\by_k;\by_h;\al,\bt)&=P_1(\by_i,\by_j,\by_k;\by_h;\al,\bt)P_2(\by_i,\by_j,\by_k;\hat{\by}_h;\al,\bt) \\ &- P_2(\by_i,\by_j,\by_k;\by_h;\al,\bt)P_1(\by_i,\by_j,\by_k;\hat{\by}_h;\al,\bt).
\end{split}
\end{equation}

\begin{prop}\label{prop:ratroot}
Let $t_a$, $a=1,2,3$ denote the three roots of the cubic equation \eqref{cubicratdef}.  The solutions $\bigl((y_l)_1,(y_l)_2\bigr)$ of the 5-point equation \eqref{n3eqsrat} are given by the six pairs of distinct roots
\begin{equation}
\bigl((y_l)_1,(y_l)_2\bigr)=(t_a,t_b),\qquad a,b=1,2,3,\quad a\neq b.
\end{equation}
\end{prop}


The proof of Proposition \ref{prop:ratroot} is omitted here for brevity, but is similar to the proof of Proposition \ref{prop:hyproot} with the role of $r_l$ and $s_l$ exchanged, and the sum of roots being equal to zero rather than the product of roots being equal to one.


\section{Lattice equations and consistency}\label{sec:latticeconsistency}

In a quasi-classical limit of the partition function, the $n$-components spin variables $\spn_i$ at vertices $i$ are effectively replaced with the classical variables $\bx_i$ at the same vertices. From the results of Section \ref{sec:QCL}, the quasi-classical expansion of the partition function \eqref{Zdef} takes the form
\begin{equation}
Z=\int_{\mathbb{R}^{n-1}}\cdots\int_{\mathbb{R}^{n-1}}\prod_{k\in V^{(int)}} d\spn_k \exp\left(-\frac{\ii}{\hbar}\mathcal{A}(\bx)+O(1)\right),
\end{equation}
where 
\begin{equation}\label{actiondef}
\begin{split}
\mathcal{A}(\bx)=\sum_{i\in V^{(B)}\cup V^{(W)}}C(\bx_i)+\sum_{(ij)\in E^{(1)}}\lag_{u_1-v_2}(\bx_i,\bx_j)+\sum_{(ij)\in E^{(2)}}\lag_{u_2-v_1}(\bx_i,\bx_j)\phantom{,} \\
+\sum_{(ij)\in E^{(3)}}\lag_{u_1-v_1}(\bx_j,\bx_i)+\sum_{(ij)\in E^{(4)}}\lag_{u_2-v_2}(\bx_j,\bx_i).
\end{split}
\end{equation}
In the quasi-classical limit $\hbar\to0$, the partition function is evaluated on the solutions of the saddle-point equations which are defined in this case by
\begin{equation}\label{eqmodef}
\frac{\partial\mathcal{A}(\bx)}{\partial (x_f)_h}=0,\qquad f\in V^{(int)},\quad h=1,\ldots,n-1.
\end{equation}
Using the change of variables of the form \eqref{hypcovdef} and taking the exponential, these equations may be put in the form
\begin{equation}\label{dlap1}
    A_i(\by_f;\by_a,\by_b,\by_c,\by_d;\al,\bt)=1,
    \qquad f\in V^{(B)},\quad i=1,\ldots,n-1,
\end{equation}
where $(fa)\in E^{(1)}$, $(fb)\in E^{(2)}$, $(fc)\in E^{(3)}$, $(fd)\in E^{(4)}$, are edges that connect nearest-neighbour vertices $a,b,c,d\in V^{(W)}$ of the vertex $f\in V^{(B)}$ and
\begin{equation}\label{dlap2}
    A_i(\by_f;\by_a,\by_c,\by_b,\by_d;\bt,\al)=1,
    \qquad f\in V^{(W)},\quad i=1,\ldots,n-1,
\end{equation}
where $(fa)\in E^{(4)}$, $(fb)\in E^{(3)}$, $(fc)\in E^{(2)}$, $(fd)\in E^{(1)}$, are edges that connect nearest-neighbour vertices $a,b,c,d\in V^{(B)}$ of the vertex $f\in V^{(W)}$.  The individual equations are defined on the 5-point configuration of vertices shown in Figure \ref{fig-dlap}.  Together the equations \eqref{dlap1} and \eqref{dlap2} define an $n$-component system of difference equations in the (rotated) square lattice.

\begin{figure}[htb]
\centering
\begin{tikzpicture}[scale=1.4]

\draw[-,thick] (-1,-1)--(1,1);
\draw[-,thick] (1,-1)--(-1,1);

\filldraw[fill=white!,draw=black!] (-1,-1) circle (1.5pt)
node[below=-8pt,right=10pt]{\color{black}$\tfrac{\alpha_1}{\beta_1}$}
node[left=1.5pt]{\color{black}$c$};
\filldraw[fill=white!,draw=black!] (1,-1) circle (1.5pt)
node[below=-8pt,left=10pt]{\color{black}$\tfrac{\alpha_1}{\beta_2}$}
node[right=1.5pt]{\color{black}$d$};
\filldraw[fill=white!,draw=black!] (-1,1) circle (1.5pt)
node[above=-8pt,right=10pt]{\color{black}$\tfrac{\alpha_2}{\beta_1}$}
node[left=1.5pt]{\color{black}$a$};
\filldraw[fill=white!,draw=black!] (1,1) circle (1.5pt)
node[above=-8pt,left=10pt]{\color{black}$\tfrac{\alpha_2}{\beta_2}$}
node[right=1.5pt]{\color{black}$b$};

\filldraw[fill=black!,draw=black!] (0,0) circle (1.5pt)
node[left=2.5pt]{\color{black}$f$};


\begin{scope}[xshift=150pt]

\draw[-,thick] (-1,-1)--(1,1);
\draw[-,thick] (1,-1)--(-1,1);

\filldraw[fill=black!,draw=black!] (-1,-1) circle (1.5pt)
node[below=-8pt,right=10pt]{\color{black}$\tfrac{\beta_2}{\alpha_2}$}
node[left=1.5pt]{\color{black}$c$};
\filldraw[fill=black!,draw=black!] (1,-1) circle (1.5pt)
node[below=-8pt,left=10pt]{\color{black}$\tfrac{\beta_1}{\alpha_2}$}
node[right=1.5pt]{\color{black}$d$};
\filldraw[fill=black!,draw=black!] (-1,1) circle (1.5pt)
node[above=-8pt,right=10pt]{\color{black}$\tfrac{\beta_2}{\alpha_1}$}
node[left=1.5pt]{\color{black}$a$};
\filldraw[fill=black!,draw=black!] (1,1) circle (1.5pt)
node[above=-8pt,left=10pt]{\color{black}$\tfrac{\beta_1}{\alpha_1}$}
node[right=1.5pt]{\color{black}$b$};

\filldraw[fill=white!,draw=black!] (0,0) circle (1.5pt)
node[left=2.5pt]{\color{black}$f$};


\end{scope}
\end{tikzpicture}

\caption{The 5-point equations \eqref{dlap1} and \eqref{dlap2}.}
\label{fig-dlap}
\end{figure}
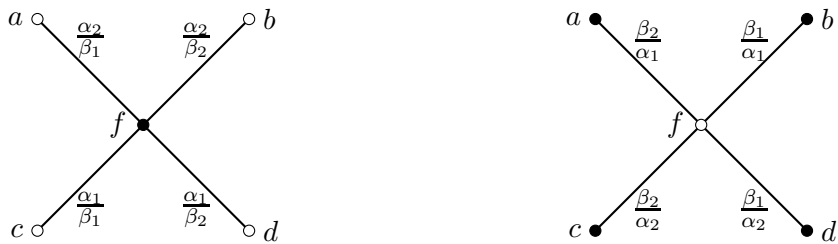

An analysis of solutions of the system of the equations \eqref{dlap1} and \eqref{dlap2} for $n>2$ goes beyond the scope of this paper. 
Numerical computations for small $n$ suggest that the individual equations may be solved for one of the four variables $\by_a,\by_b,\by_c,\by_d$ in terms of the others and the solution is unique up to permutations of the components, as was seen in the previous section for $n=2$ and $n=3$.  Thus, one may interpret \eqref{dlap1} and \eqref{dlap2} as a system of multicomponent evolution equations in the square lattice, analogously to the scalar $n=2$ case \cite{Kels:2020zjn}.  Typical examples of initial conditions are shown in Figure \ref{fig-initial} for evolutions in the north-east direction of the lattice (see also \cite{GubbiottiKels} for more general types of applicable initial conditions).

Note that in the context of variational principles, the system of equations \eqref{dlap1} and \eqref{dlap2} may be regarded as coming from the equations of motion \eqref{eqmodef} for the action $\mathcal{A}(\bx)$, while the partition function provides a natural quantization of the latter action.  Note also that here only the leading order $O(\hbar^{-1})$ is being considered, and it is likely that other potentially interesting equations can be found at subleading orders of the quasi-classical expansion.  However, determining further orders of the expansion would require more complicated computations than have been considered here and is beyond the scope of this paper.

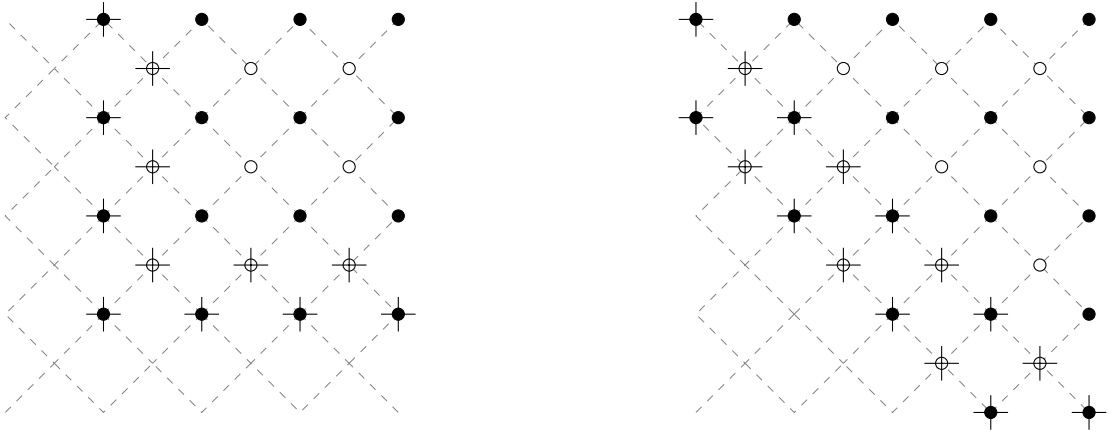
\begin{figure}[htb]
\centering
\begin{tikzpicture}[scale=0.65]

\draw[-,gray,very thin,dashed] (0,0)--(-1,-1);
\draw[-,gray,very thin,dashed] (0,0)--(-1,1);
\draw[-,gray,very thin,dashed] (0,0)--(1,1);
\draw[-,gray,very thin,dashed] (0,0)--(1,-1);

\draw[-,gray,very thin,dashed] (-2,-2)--(-3,-3);
\draw[-,gray,very thin,dashed] (-2,-2)--(-3,-1);
\draw[-,gray,very thin,dashed] (-2,-2)--(-1,-1);
\draw[-,gray,very thin,dashed] (-2,-2)--(-1,-3);

\draw[-,gray,very thin,dashed] (-2,0)--(-3,-1);
\draw[-,gray,very thin,dashed] (-2,0)--(-3,1);
\draw[-,gray,very thin,dashed] (-2,0)--(-1,1);
\draw[-,gray,very thin,dashed] (-2,0)--(-1,-1);

\draw[-,gray,very thin,dashed] (-2,2)--(-3,1);
\draw[-,gray,very thin,dashed] (-2,2)--(-3,3);
\draw[-,gray,very thin,dashed] (-2,2)--(-1,3);
\draw[-,gray,very thin,dashed] (-2,2)--(-1,1);

\draw[-,gray,very thin,dashed] (0,2)--(-1,1);
\draw[-,gray,very thin,dashed] (0,2)--(-1,3);
\draw[-,gray,very thin,dashed] (0,2)--(1,3);
\draw[-,gray,very thin,dashed] (0,2)--(1,1);

\draw[-,gray,very thin,dashed] (2,2)--(1,1);
\draw[-,gray,very thin,dashed] (2,2)--(1,3);
\draw[-,gray,very thin,dashed] (2,2)--(3,3);
\draw[-,gray,very thin,dashed] (2,2)--(3,1);

\draw[-,gray,very thin,dashed] (2,0)--(1,-1);
\draw[-,gray,very thin,dashed] (2,0)--(1,1);
\draw[-,gray,very thin,dashed] (2,0)--(3,1);
\draw[-,gray,very thin,dashed] (2,0)--(3,-1);

\draw[-,gray,very thin,dashed] (2,-2)--(1,-3);
\draw[-,gray,very thin,dashed] (2,-2)--(1,-1);
\draw[-,gray,very thin,dashed] (2,-2)--(3,-1);
\draw[-,gray,very thin,dashed] (2,-2)--(3,-3);

\draw[-,gray,very thin,dashed] (0,-2)--(-1,-3);
\draw[-,gray,very thin,dashed] (0,-2)--(-1,-1);
\draw[-,gray,very thin,dashed] (0,-2)--(1,-1);
\draw[-,gray,very thin,dashed] (0,-2)--(1,-3);

\draw[-,gray,very thin,dashed] (-4,-4)--(-3,-3)--(-4,-2)--(-3,-1)--(-4,0)--(-3,1)--(-4,2)--(-3,3)--(-4,4);
\draw[-,gray,very thin,dashed]          (-3,-3)--(-2,-4)--(-1,-3)--(0,-4)--(1,-3)--(2,-4)--(3,-3)--(4,-4);
\draw[-,gray,very thin,dashed] (4,4)--(3,3)--(4,2)--(3,1)--(4,0)--(3,-1)--(4,-2)--(3,-3);
\draw[-,gray,very thin,dashed]        (3,3)--(2,4)--(1,3)--(0,4)--(-1,3)--(-2,4)--(-3,3);


\filldraw[draw=black,fill=black] (-2,-2) circle (3.5pt)
node[cross=5pt,rotate=45]{};

\filldraw[draw=black,fill=black] (-2,-0) circle (3.5pt)
node[cross=5pt,rotate=45]{};
\filldraw[draw=black,fill=black] (0,-2) circle (3.5pt)
node[cross=5pt,rotate=45]{};
\filldraw[draw=black,fill=white] (-1,-1) circle (3.5pt)
node[cross=5pt,rotate=45]{};

\filldraw[draw=black,fill=black] (0,-0) circle (3.5pt);
\filldraw[draw=black,fill=black] (-2,2) circle (3.5pt)
node[cross=5pt,rotate=45]{};
\filldraw[draw=black,fill=black] (2,-2) circle (3.5pt)
node[cross=5pt,rotate=45]{};
\filldraw[draw=black,fill=white] (-1,1) circle (3.5pt)
node[cross=5pt,rotate=45]{};
\filldraw[draw=black,fill=white] (1,-1) circle (3.5pt)
node[cross=5pt,rotate=45]{};

\filldraw[draw=black,fill=black] (-2,4) circle (3.5pt)
node[cross=5pt,rotate=45]{};
\filldraw[draw=black,fill=black] (-0,2) circle (3.5pt);
\filldraw[draw=black,fill=black] (2,-0) circle (3.5pt);
\filldraw[draw=black,fill=black] (4,-2) circle (3.5pt)
node[cross=5pt,rotate=45]{};
\filldraw[draw=black,fill=white] (-1,3) circle (3.5pt)
node[cross=5pt,rotate=45]{};
\filldraw[draw=black,fill=white] (1,1) circle (3.5pt);
\filldraw[draw=black,fill=white] (3,-1) circle (3.5pt)
node[cross=5pt,rotate=45]{};

\filldraw[draw=black,fill=black] (0,4) circle (3.5pt);
\filldraw[draw=black,fill=black] (2,2) circle (3.5pt);
\filldraw[draw=black,fill=black] (4,-0) circle (3.5pt);
\filldraw[draw=black,fill=white] (1,3) circle (3.5pt);
\filldraw[draw=black,fill=white] (3,1) circle (3.5pt);

\filldraw[draw=black,fill=black] (2,4) circle (3.5pt);
\filldraw[draw=black,fill=black] (4,2) circle (3.5pt);
\filldraw[draw=black,fill=white] (3,3) circle (3.5pt);

\filldraw[draw=black,fill=black] (4,4) circle (3.5pt);

\begin{scope}[xshift=400]

\draw[-,gray,very thin,dashed] (0,0)--(-1,-1);
\draw[-,gray,very thin,dashed] (0,0)--(-1,1);
\draw[-,gray,very thin,dashed] (0,0)--(1,1);
\draw[-,gray,very thin,dashed] (0,0)--(1,-1);

\draw[-,gray,very thin,dashed] (-2,-2)--(-3,-3);
\draw[-,gray,very thin,dashed] (-2,-2)--(-3,-1);
\draw[-,gray,very thin,dashed] (-2,-2)--(-1,-1);
\draw[-,gray,very thin,dashed] (-2,-2)--(-1,-3);

\draw[-,gray,very thin,dashed] (-2,0)--(-3,-1);
\draw[-,gray,very thin,dashed] (-2,0)--(-3,1);
\draw[-,gray,very thin,dashed] (-2,0)--(-1,1);
\draw[-,gray,very thin,dashed] (-2,0)--(-1,-1);

\draw[-,gray,very thin,dashed] (-2,2)--(-3,1);
\draw[-,gray,very thin,dashed] (-2,2)--(-3,3);
\draw[-,gray,very thin,dashed] (-2,2)--(-1,3);
\draw[-,gray,very thin,dashed] (-2,2)--(-1,1);

\draw[-,gray,very thin,dashed] (0,2)--(-1,1);
\draw[-,gray,very thin,dashed] (0,2)--(-1,3);
\draw[-,gray,very thin,dashed] (0,2)--(1,3);
\draw[-,gray,very thin,dashed] (0,2)--(1,1);

\draw[-,gray,very thin,dashed] (2,2)--(1,1);
\draw[-,gray,very thin,dashed] (2,2)--(1,3);
\draw[-,gray,very thin,dashed] (2,2)--(3,3);
\draw[-,gray,very thin,dashed] (2,2)--(3,1);

\draw[-,gray,very thin,dashed] (2,0)--(1,-1);
\draw[-,gray,very thin,dashed] (2,0)--(1,1);
\draw[-,gray,very thin,dashed] (2,0)--(3,1);
\draw[-,gray,very thin,dashed] (2,0)--(3,-1);

\draw[-,gray,very thin,dashed] (2,-2)--(1,-3);
\draw[-,gray,very thin,dashed] (2,-2)--(1,-1);
\draw[-,gray,very thin,dashed] (2,-2)--(3,-1);
\draw[-,gray,very thin,dashed] (2,-2)--(3,-3);

\draw[-,gray,very thin,dashed] (0,-2)--(-1,-3);
\draw[-,gray,very thin,dashed] (0,-2)--(-1,-1);
\draw[-,gray,very thin,dashed] (0,-2)--(1,-1);
\draw[-,gray,very thin,dashed] (0,-2)--(1,-3);

\draw[-,gray,very thin,dashed] (-4,2)--(-3,3)--(-4,4);\draw[-,gray,very thin,dashed] (2,-4)--(3,-3)--(4,-4);

\draw[-,gray,very thin,dashed] (-4,-4)--(-3,-3)--(-4,-2)--(-3,-1)--(-4,0)--(-3,1)--(-4,2);
\draw[-,gray,very thin,dashed]          (-3,-3)--(-2,-4)--(-1,-3)--(0,-4)--(1,-3)--(2,-4);
\draw[-,gray,very thin,dashed] (4,4)--(3,3)--(4,2)--(3,1)--(4,0)--(3,-1)--(4,-2)--(3,-3);
\draw[-,gray,very thin,dashed]        (3,3)--(2,4)--(1,3)--(0,4)--(-1,3)--(-2,4)--(-3,3);

\filldraw[draw=black,fill=black] (-4,2) circle (3.5pt)
node[cross=5pt,rotate=45]{};
\filldraw[draw=black,fill=black] (-2,0) circle (3.5pt)
node[cross=5pt,rotate=45]{};
\filldraw[draw=black,fill=black] (-0,-2) circle (3.5pt)
node[cross=5pt,rotate=45]{};
\filldraw[draw=black,fill=black] (2,-4) circle (3.5pt)
node[cross=5pt,rotate=45]{};
\filldraw[draw=black,fill=white] (-3,1) circle (3.5pt)
node[cross=5pt,rotate=45]{};
\filldraw[draw=black,fill=white] (-1,-1) circle (3.5pt)
node[cross=5pt,rotate=45]{};
\filldraw[draw=black,fill=white] (1,-3) circle (3.5pt)
node[cross=5pt,rotate=45]{};

\filldraw[draw=black,fill=black] (-4,4) circle (3.5pt)
node[cross=5pt,rotate=45]{};
\filldraw[draw=black,fill=black] (-2,2) circle (3.5pt)
node[cross=5pt,rotate=45]{};
\filldraw[draw=black,fill=black] (-0,0) circle (3.5pt)
node[cross=5pt,rotate=45]{};
\filldraw[draw=black,fill=black] (2,-2) circle (3.5pt)
node[cross=5pt,rotate=45]{};
\filldraw[draw=black,fill=black] (4,-4) circle (3.5pt)
node[cross=5pt,rotate=45]{};
\filldraw[draw=black,fill=white] (-3,3) circle (3.5pt)
node[cross=5pt,rotate=45]{};
\filldraw[draw=black,fill=white] (-1,1) circle (3.5pt)
node[cross=5pt,rotate=45]{};
\filldraw[draw=black,fill=white] (1,-1) circle (3.5pt)
node[cross=5pt,rotate=45]{};
\filldraw[draw=black,fill=white] (3,-3) circle (3.5pt)
node[cross=5pt,rotate=45]{};

\filldraw[draw=black,fill=black] (-2,4) circle (3.5pt);
\filldraw[draw=black,fill=black] (-0,2) circle (3.5pt);
\filldraw[draw=black,fill=black] (2,-0) circle (3.5pt);
\filldraw[draw=black,fill=black] (4,-2) circle (3.5pt);
\filldraw[draw=black,fill=white] (-1,3) circle (3.5pt);
\filldraw[draw=black,fill=white] (1,1) circle (3.5pt);
\filldraw[draw=black,fill=white] (3,-1) circle (3.5pt);

\filldraw[draw=black,fill=black] (0,4) circle (3.5pt);
\filldraw[draw=black,fill=black] (2,2) circle (3.5pt);
\filldraw[draw=black,fill=black] (4,-0) circle (3.5pt);
\filldraw[draw=black,fill=white] (1,3) circle (3.5pt);
\filldraw[draw=black,fill=white] (3,1) circle (3.5pt);

\filldraw[draw=black,fill=black] (2,4) circle (3.5pt);
\filldraw[draw=black,fill=black] (4,2) circle (3.5pt);
\filldraw[draw=black,fill=white] (3,3) circle (3.5pt);

\filldraw[draw=black,fill=black] (4,4) circle (3.5pt);







\end{scope}

\end{tikzpicture}
\caption{Corner- and staircase-type initial conditions indicated by crosses.}
\label{fig-initial}
\end{figure}

\subsection{ IRF Yang-Baxter equation and consistency}
For the $n=2$ case, the equations \eqref{dlap1} and \eqref{dlap2} were shown to be consistent as an overdetermined system of 14 equations defined on vertices of a face-centered cubic unit cell \cite{Kels:2020zjn}.  The extension of the system of 14 equations for $n>2$ may be deduced from the quasi-classical expansion of the IRF form of the YBE given in \eqref{YBE-IRF}.  First, it is convenient to relabel the variables in \eqref{YBE-IRF} as
\begin{equation}
\spn_c\to\spn_{a'},\;\spn_d\to\spn_c,\;\spn_e\to\spn_{c'},\;\spn_f\to\spn_d,\;\spn_g\to\spn_{b'},
\end{equation}
and denote $\bp=(p_1,p_2),\bq=(q_1,q_2),\br=(r_1,r_2)$.  Then the IRF YBE may be written in terms of $V^{(B)}$ as
\begin{equation}\label{YBEcov}
\begin{split}
\int_{\mathbb{R}^{n-1}}d\spn_aS(\spn_a) X_\bq(\spn_c,\spn_{c'},\spn_a,\spn_{a'})
V^{(B)}_{\bp\bq}(\spn_{a'},\spn_a,\spn_{c'},\spn_c)
V^{(B)}_{\bp\br}(\spn_a,\spn_b,\spn_c,\spn_d)
V^{(B)}_{\bq\br}(\spn_{a'},\spn_{b'},\spn_a,\spn_b)=  \\
\int_{\mathbb{R}^{n-1}}d\spn_{d'}S(\spn_{d'}) X_\bq(\spn_d,\spn_{d'},\spn_b,\spn_{b'})
V^{(B)}_{\bp\bq}(\spn_{b'},\spn_b,\spn_{d'},\spn_d)
V^{(B)}_{\bp\br}(\spn_{a'},\spn_{b'},\spn_{c'},\spn_{d'})
V^{(B)}_{\bq\br}(\spn_{c'},\spn_{d'},\spn_c,\spn_d),
\end{split}
\end{equation}
where 
\begin{equation}
X_\bq(\spn_c,\spn_{c'},\spn_a,\spn_{a'})=W_{q_2-q_1}(\spn_c,\spn_{c'})W_{q_1-q_2}(\spn_{a'},\spn_a).
\end{equation}
Taking a quasi-classical expansion as in \eqref{qclas}, along with $r_j\to \frac{w_j}{\sqrt{2\pi\hbar}}$, gives asymptotics of \eqref{YBEcov} of the form
\begin{equation}\label{YBEqcl}
\begin{split}
\int_{\mathbb{R}^{n-1}}\cdots\int_{\mathbb{R}^{n-1}}d\spn_ad\spn_gd\spn_ed\spn_f\exp\left(-\frac{\ii}{\hbar}\mathcal{A}^{(L)}(\bx)+O(1)\right)\phantom{,} \\
=\int_{\mathbb{R}^{n-1}}\cdots\int_{\mathbb{R}^{n-1}}d\spn_{d'}d\spn_{g'}d\spn_{e'}d\spn_{f'}\exp\left(-\frac{\ii}{\hbar}\mathcal{A}^{(R)}(\bx)+O(1)\right),
\end{split}
\end{equation}
where
\begin{equation}
\begin{split}
\mathcal{A}^{(L)}(\bx)
=\lag^{(B)}_{\bu\bv}(\bx_g;\bx_{a'},\bx_a,\bx_{c'},\bx_c)
+\lag^{(B)}_{\bu\bw}(\bx_e;\bx_a,\bx_b,\bx_c,\bx_d)
+\lag^{(B)}_{\bv\bw}(\bx_f;\bx_{a'},\bx_{b'},\bx_a,\bx_b)\phantom{,} \\
+\lag_{v_2-v_1}(\bx_c,\bx_{c'})+\lag_{v_1-v_2}(\bx_{a'},\bx_a), \\
\mathcal{A}^{(R)}(\bx)
=\lag^{(B)}_{\bu\bv}(\bx_{g'};\bx_{b'},\bx_b,\bx_{d'},\bx_d)
+\lag^{(B)}_{\bu\bw}(\bx_{e'};\bx_{a'},\bx_{b'},\bx_{c'},\bx_{d'})
+\lag^{(B)}_{\bv\bw}(\bx_{f'};\bx_{c'},\bx_{d'},\bx_c,\bx_d)\phantom{,} \\
+\lag_{v_2-v_1}(\bx_d,\bx_{d'})+\lag_{v_1-v_2}(\bx_{b'},\bx_b),
\end{split}
\end{equation}
and $\bu=(u_1,u_2)$, $\bv=(v_1,v_2)$, and $\bw=(w_1,w_2)$.  The saddle-point equations for the integrals of \eqref{YBEqcl} may be written in the combined form
\begin{equation}
\begin{gathered}
\frac{\partial\bigl(\mathcal{A}^{(L)}(\bx)-\mathcal{A}^{(R)}(\bx)\bigr)}{\partial (\cbx_\delta)_i}=0,\quad i=1,\ldots,n-1,\; \delta=a,e,f,g,d'e'f'g'.
\end{gathered}
\end{equation}
Using a change of variables as given in \eqref{hypcovdef}, but with a slightly different change of variables for the parameters
\begin{equation}
\alpha_j=\EXP^{\ii u_j},\quad \beta_j=\EXP^{\ii w_j},\quad \gamma_j=\EXP^{\ii v_j}, \qquad j=1,2,
\end{equation}
the exponentials of the above saddle-point equations may be written in the form
\begin{equation}\label{cafcc8}\begin{split} 
A_i(\by_g;\by_{a'},\by_a,\by_{c'},\by_c;\al,\gm)=1,\quad &
A_i(\by_{g'};\by_{g'};\by_{b'},\by_b,\by_{d'};\al,\gm)=1, \\
A_i(\by_e;\by_a,\by_b,\by_c,\by_d;\al,\bt)=1, \quad &
A_i(\by_{e'};\by_{a'},\by_{b'},\by_{c'},\by_{d'};\al,\bt)=1, \\
A_i(\by_f;\by_{a'},\by_{b'},\by_a,\by_b;\gm,\bt)=1, \quad &
A_i(\by_{f'};\by_{c'},\by_{d'},\by_c,\by_d;\gm,\bt)=1, \\
A_i(\by_a;\by_g,\by_{a'},\by_e,\by_f;(\beta_1,\gamma_2),(\alpha_2,\gamma_1))=1, \quad &
A_i(\by_{d'};\by_{g'},\by_{d},\by_{e'},\by_{f'};(\beta_2,\gamma_1),(\alpha_1,\gamma_2))=1,
\end{split}\end{equation}
for $i=1,\ldots,n-1$.  Equating both sides of \eqref{YBEqcl} at leading order $O(\hbar^{-1})$ gives
\begin{equation}
\mathcal{A}^{(L)}=\mathcal{A}^{(R)},
\end{equation}
which is assumed to hold on solutions of the saddle-point equations \eqref{cafcc8}.  The latter equation may be differentiated with respect to components of the six variables $\bx_b,\bx_c,\bx_d,\bx_{b'},\bx_{c'},\bx_{a'}$, which results in a further six equations
\begin{equation}\label{cafcc6}\begin{split}
A_i(\by_b;\by_{g'},\by_{b'},\by_e,\by_f;(\beta_2,\gamma_2),(\alpha_2,\gamma_1))=1, \quad &
A_i(\by_{b'};\by_{g'},\by_b,\by_{e'},\by_f;(\beta_2,\gamma_1),(\alpha_2,\gamma_2))=1, \\ 
A_i(\by_{c};\by_g,\by_{c'},\by_e,\by_{f'};(\beta_1,\gamma_2),(\alpha_1,\gamma_1))=1, \quad &
A_i(\by_{c'};\by_g,\by_c,\by_{e'},\by_{f'};(\beta_1,\gamma_1),(\alpha_1,\gamma_2))=1, \\
A_i(\by_d;\by_{g'},\by_{d'},\by_e,\by_{f'};(\beta_2,\gamma_2),(\alpha_1,\gamma_1))=1, \quad &
A_i(\by_{a'};\by_g,\by_a,\by_{e'},\by_f;(\beta_1,\gamma_1),(\alpha_2,\gamma_2))=1,
\end{split}\end{equation}
for $i=1,\ldots,n-1$.  

The 14 equations in \eqref{cafcc8} and \eqref{cafcc6} are the $n-1$-component analogues of the fourteen scalar equations that were used for the formulation of consistency-around-a-face-centered-cube (CAFCC) in \cite{Kels:2020zjn}.  
Specifically, the 14 equations are assigned to the face-centered cubic unit cell shown in Figure \ref{fig-CAFCC}, and CAFCC requires that the fourteen equations in \eqref{cafcc8} and \eqref{cafcc6} are consistent.  Furthermore, for scalar equations ($n=2$) there is a procedure to derive Lax pairs from a consistent set of CAFCC equations \cite{KelsLax1,KelsLax2}.  Similarly to \cite{Kels:2020zjn}, one way to check CAFCC is to fix the six $n-1$-component variables $\by_a,\by_b,\by_c,\by_e,\by_f,\by_g$, and the fourteen equations should agree for solutions of the remaining eight variables.  An analytic verification of consistency for $n>2$ is beyond the scope of this paper, but the $n-1$-component systems of equations \eqref{fcqhypa} and \eqref{fcqrata} have been numerically checked to be consistent for $n=3,4,5$.

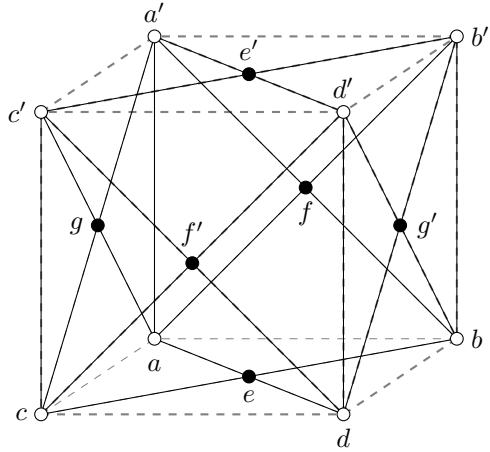
\begin{figure}[hbt!]
\centering
\begin{tikzpicture}[scale=1.0]

\draw[-,gray,thick,dashed] (2.5,4)--(-1.5,4)--(-3,3)--(1,3)--(2.5,4);
\draw[-,gray,thick,dashed] (-3,-1)--(1,-1)--(2.5,0);\draw[-,gray,very thin,dashed] (2.5,0)--(-1.5,0)--(-3,-1);

\draw[-,gray,very thin,dashed] (-3,-1)--(2.5,0);\draw[-,gray,very thin,dashed] (-1.5,0)--(1,-1);
\draw[-,gray,thick,dashed] (-3,3)--(-3,-1);\draw[-,gray,very thin,dashed] (-3,-1)--(-1.5,4);\draw[-,gray,very thin,dashed] (-3,3)--(-1.5,0)--(-1.5,4);
\draw[-,gray,very thin,dashed] (-1.5,0)--(2.5,4);\draw[-,gray,very thin,dashed] (-1.5,4)--(2.5,0);

\draw[-,gray,thick,dashed] (-3,-1)--(1,3);\draw[-,gray,thick,dashed] (-3,3)--(1,-1);
\draw[-,gray,thick,dashed] (1,-1)--(2.5,4)--(2.5,0);\draw[-,gray,thick,dashed] (1,-1)--(1,3)--(2.5,0);
\draw[-,gray,thick,dashed] (-3,3)--(2.5,4);\draw[-,gray,thick,dashed] (-1.5,4)--(1,3);

\draw[-] (-3,-1)--(2.5,0);\draw[-] (1,-1)--(-1.5,0);
\draw[-] (-3,3)--(-3,-1)--(-1.5,4);\draw[-] (-3,3)--(-1.5,0)--(-1.5,4);
\draw[-] (-1.5,0)--(2.5,4);\draw[-] (-1.5,4)--(2.5,0);

\draw[-] (-3,-1)--(1,3);\draw[-] (-3,3)--(1,-1);
\draw[-] (1,-1)--(2.5,4)--(2.5,0);\draw[-] (1,-1)--(1,3)--(2.5,0);
\draw[-] (-3,3)--(2.5,4);\draw[-] (-1.5,4)--(1,3);

\filldraw[fill=white,draw=black] (1,-1) circle (2.4pt) node[below=1.5pt]{\small $d$};
\filldraw[fill=black,draw=black] (-0.25,-0.5) circle (2.4pt) node[below=1.5pt]{\small $e$};
\filldraw[fill=white,draw=black] (-1.5,0) circle (2.4pt) node[below=4pt]{\color{black}\small $a$};
\filldraw[fill=white,draw=black] (-3,-1) circle (2.4pt) node[left=1.5pt]{\small $c$};
\filldraw[fill=white,draw=black] (2.5,0) circle (2.4pt) node[right=1.5pt]{\color{black}\small $b$};
\filldraw[fill=white,draw=black] (-3,3) circle (2.4pt) node[left=1.5pt]{\color{black}\small $c'$};
\filldraw[fill=white,draw=black] (-1.5,4) circle (2.4pt) node[above=1.5pt]{\small $a'$};
\filldraw[fill=black,draw=black] (-2.25,1.5) circle (2.4pt) node[left=1.5pt]{\color{black}\small $g$};
\filldraw[fill=black,draw=black] (0.5,2) circle (2.4pt) node[below=1.5pt]{\small $f$};
\filldraw[fill=white,draw=black] (2.5,4) circle (2.4pt) node[right=1.5pt]{\small $b'$};

\filldraw[fill=white,draw=black] (1,3) circle (2.4pt) node[above=1.5pt]{\small $d'$};
\filldraw[fill=black,draw=black] (-0.25,3.5) circle (2.4pt) node[above=1.5pt]{\small $e'$};
\filldraw[fill=black,draw=black] (-1.0,1) circle (2.4pt) node[above=2.5pt]{\small $f'$};
\filldraw[fill=black,draw=black] (1.75,1.5) circle (2.4pt) node[right=2.5pt]{\small $g'$};




\end{tikzpicture}
\caption{The fourteen equations \eqref{cafcc8} and \eqref{cafcc6}. }
\label{fig-CAFCC}
\end{figure}

 \section{Conclusion}

This paper has investigated the quasi-classical expansion of a particular hyperbolic solution of the star-star relation that is associated to multivariate hyperbolic hypergeometric functions.  Difference equations that were derived from the quasi-classical expansion were shown to provide multicomponent extensions of known integrable scalar 5-point difference equations that were previously studied in the context of consistency on face-centered cubics \cite{Kels:2020zjn}.  Integrability of these 5-point multicomponent difference equations was proposed in the form of the consistency-around-a-face-centered-cubic property, which is essentially a reformulation of equations that arise in the quasi-classical expansion of the IRF Yang-Baxter equation associated to the edge-interaction lattice model, as shown in Section \ref{sec:latticeconsistency}.  

These results build on the correspondence that has been developed \cite{Bazhanov:2007mh,Bazhanov:2010kz,Bazhanov:2011mz,Bazhanov:2016ajm,Kels:2018xge,Kels:2020zjn}  between integrable lattice models of statistical mechanics and integrable systems of difference equations in the quasi-classical limit.  These are two types of models which are classified as integrable but apparently have completely different characteristics that describe their integrability.  The main importance of investigating the quasi-classical limit, as has been done in this paper, is that it shows how these independent classes of integrable models are in fact intrinsically connected, and it provides a new way to further develop our understanding of both types of integrable models.  Specifically, this has been realised in this paper through the connection of the multicomponent star-star relation for integrable lattice models of statistical mechanics and new 5-point difference equations which provide multicomponent extensions of known integrable scalar difference equations.

For future research directions, it would be interesting to be able to extend the known integrable characteristics of the scalar $n=2$ equations to the case of $n>2$, where possible.  For example, this includes understanding the solutions of the equations, and showing that CAFCC is satisfied for $n>2$ and deriving associated Lax pairs.  It would also be interesting to study the hex systems \cite{Kels:2022pdz,GKV24} associated to the multicomponent 5-point equations.  In the scalar case, the 5-point equations may be constructed using three-leg forms of quad equations \cite{ABS}, while the same 4-point quad equations are contained in the expressions for the 5-point equations themselves.  However, the analogue of such constructions are not yet understood for $n>2$, and the relevant multicomponent 4-point equations (if they exist) have not appeared in the literature, as far as the author is aware.  Finally, investigating any integrable reductions of the multicomponent equations would also be of interest.  

\section*{Acknowledgements}

This research was partially supported through ARC Discovery Project DP200102118.

{\small
\bibliography{MComp}

\providecommand{\href}[2]{#2}\begingroup\raggedright\begin{thebibliography}{10}

\bibitem{Baxter:1982zz}
R.~J. Baxter, {\em Exactly {S}olved {M}odels in {S}tatistical {M}echanics}.
\newblock Academic, London,
1982.
\newblock

\bibitem{McCoyBook}
B.~M. McCoy, {\em Advanced statistical mechanics}, vol.~146 of {\em International Series of Monographs on Physics}.
\newblock Oxford University Press, New York, paperback~ed., 2015.

\bibitem{PerkSTR}
H.~Au-Yang and J.~H.~H. Perk, \href{http://dx.doi.org/10.2969/aspm/01910057}{``Onsager's star-triangle equation: master key to integrability,''} in {\em Integrable systems in quantum field theory and statistical mechanics}, vol.~19 of {\em Adv. Stud. Pure Math.}, pp.~57--94.
\newblock Academic Press, Boston, MA, 1989.

\bibitem{PerkYBEs}
J.~Perk and H.~Au-Yang, \href{http://dx.doi.org/https://doi.org/10.1016/B0-12-512666-2/00191-7}{``{Y}ang–-{B}axter equations,''} in {\em Encyclopedia of Mathematical Physics}, J.-P. Françoise, G.~L. Naber, and T.~S. Tsun, eds., pp.~465 -- 473.
\newblock Academic Press, Oxford, 2006.

\bibitem{Bazhanov:2016ajm}
V.~V. Bazhanov, A.~P. Kels, and S.~M. Sergeev, ``{Quasi-classical expansion of the star-triangle relation and integrable systems on quad-graphs},'' \href{http://dx.doi.org/10.1088/1751-8113/49/46/464001}{{\em J. Phys.} {\bfseries A49} (2016) 464001},
\href{http://arxiv.org/abs/1602.07076}{{\ttfamily arXiv:1602.07076 [math-ph]}}.

\bibitem{Fateev:1982wi}
V.~A. Fateev and A.~B. Zamolodchikov, ``{Selfdual solutions of the star triangle relations in Z(N) models},''
\href{http://dx.doi.org/10.1016/0375-9601(82)90736-8}{{\em Phys. Lett.} {\bfseries A92} (1982) 37--39}.

\bibitem{Kashiwara:1986tu}
M.~Kashiwara and T.~Miwa, ``{A Class of Elliptic Solutions to the Star Triangle Relation},''
\href{http://dx.doi.org/10.1016/0550-3213(86)90591-2}{{\em Nucl. Phys.} {\bfseries B275} (1986) 121}.

\bibitem{AuYang:1987zc}
H.~Au-Yang, B.~M. McCoy, J.~H.~H. Perk, S.~Tang, and M.-L. Yan, ``{Commuting transfer matrices in the chiral Potts models: Solutions of Star triangle equations with genus $> 1$},''
\href{http://dx.doi.org/10.1016/0375-9601(87)90065-X}{{\em Phys. Lett.} {\bfseries A123} (1987) 219--223}.

\bibitem{Baxter:1987eq}
R.~J. Baxter, J.~H.~H. Perk, and H.~Au-Yang, ``{New solutions of the star triangle relations for the chiral Potts model},''
\href{http://dx.doi.org/10.1016/0375-9601(88)90896-1}{{\em Phys. Lett.} {\bfseries A128} (1988) 138--142}.

\bibitem{Bazhanov:2007mh}
V.~V. Bazhanov, V.~V. Mangazeev, and S.~M. Sergeev, ``{Faddeev-Volkov solution of the {Y}ang-{B}axter equation and discrete conformal symmetry},'' \href{http://dx.doi.org/10.1016/j.nuclphysb.2007.05.013}{{\em Nucl. Phys.} {\bfseries B784} (2007) 234--258},
\href{http://arxiv.org/abs/hep-th/0703041}{{\ttfamily arXiv:hep-th/0703041 [hep-th]}}.

\bibitem{Bazhanov:2010kz}
V.~V. Bazhanov and S.~M. Sergeev, ``{A Master solution of the quantum {Y}ang-{B}axter equation and classical discrete integrable equations},'' \href{http://dx.doi.org/10.4310/ATMP.2012.v16.n1.a3}{{\em Adv. Theor. Math. Phys.} {\bfseries 16} no.~1, (2012) 65--95},
\href{http://arxiv.org/abs/1006.0651}{{\ttfamily arXiv:1006.0651 [math-ph]}}.

\bibitem{Spiridonov:2010em}
V.~P. Spiridonov, ``{Elliptic beta integrals and solvable models of statistical mechanics},'' \href{http://dx.doi.org/10.1090/conm/563}{{\em Contemp. Math.} {\bfseries 563} (2012) 181--211},
\href{http://arxiv.org/abs/1011.3798}{{\ttfamily arXiv:1011.3798 [hep-th]}}.

\bibitem{Zabrodin:2010qm}
A.~Zabrodin, ``{Intertwining operators for Sklyanin algebra and elliptic hypergeometric series},'' \href{http://dx.doi.org/10.1016/j.geomphys.2011.02.019}{{\em J. Geom. Phys.} {\bfseries 61} (2011) 1733--1754},
\href{http://arxiv.org/abs/1012.1228}{{\ttfamily arXiv:1012.1228 [math-ph]}}.

\bibitem{Derkachov:2012iv}
S.~E. Derkachov and V.~P. Spiridonov, ``{Yang-Baxter equation, parameter permutations, and the elliptic beta integral},'' \href{http://dx.doi.org/10.1070/RM2013v068n06ABEH004869}{{\em Russ. Math. Surveys} {\bfseries 68} (2013) 1027--1072}, \href{http://arxiv.org/abs/1205.3520}{{\ttfamily arXiv:1205.3520 [math-ph]}}.

\bibitem{Chicherin:2012yn}
D.~Chicherin, S.~Derkachov, and A.~P. Isaev, ``{Conformal group: R-matrix and star-triangle relation},'' \href{http://dx.doi.org/10.1007/JHEP04(2013)020}{{\em JHEP} {\bfseries 04} (2013) 020}, \href{http://arxiv.org/abs/1206.4150}{{\ttfamily arXiv:1206.4150 [math-ph]}}.

\bibitem{Yamazaki:2013nra}
M.~Yamazaki, ``{New Integrable Models from the Gauge/YBE Correspondence},'' \href{http://dx.doi.org/10.1007/s10955-013-0884-8}{{\em J. Statist. Phys.} {\bfseries 154} (2014) 895},
\href{http://arxiv.org/abs/1307.1128}{{\ttfamily arXiv:1307.1128 [hep-th]}}.

\bibitem{Kels:2015bda}
A.~P. Kels, ``{New solutions of the star--triangle relation with discrete and continuous spin variables},'' \href{http://dx.doi.org/10.1088/1751-8113/48/43/435201}{{\em J. Phys.} {\bfseries A48} no.~43, (2015) 435201},
\href{http://arxiv.org/abs/1504.07074}{{\ttfamily arXiv:1504.07074 [math-ph]}}.

\bibitem{Kashaev:2015nya}
R.~Kashaev, ``{The Yang-Baxter relation and gauge invariance},'' \href{http://dx.doi.org/10.1088/1751-8113/49/16/164001}{{\em J. Phys.} {\bfseries A49} no.~16, (2016) 164001},
\href{http://arxiv.org/abs/1510.03043}{{\ttfamily arXiv:1510.03043 [math-ph]}}.

\bibitem{Gahramanov:2016wxi}
I.~Gahramanov and H.~Rosengren, ``{Basic hypergeometry of supersymmetric dualities},'' \href{http://dx.doi.org/10.1016/j.nuclphysb.2016.10.004}{{\em Nucl. Phys. B} {\bfseries 913} (2016) 747--768}, \href{http://arxiv.org/abs/1606.08185}{{\ttfamily arXiv:1606.08185 [hep-th]}}.

\bibitem{GahramanovKels}
I.~Gahramanov and A.~P. Kels, ``{The star-triangle relation, lens partition function, and hypergeometric sum/integrals},'' \href{http://dx.doi.org/10.1007/JHEP02(2017)040}{{\em JHEP} {\bfseries 02} (2017) 040},
\href{http://arxiv.org/abs/1610.09229}{{\ttfamily arXiv:1610.09229 [math-ph]}}.

\bibitem{Kels:2018xge}
A.~P. Kels, ``{Integrable quad equations derived from the quantum {Y}ang-{B}axter equation},'' \href{http://dx.doi.org/10.1007/s11005-020-01255-3}{{\em Lett. Math. Phys.} {\bfseries 110} (2020) 1477--1557},
\href{http://arxiv.org/abs/1803.03219}{{\ttfamily arXiv:1803.03219 [math-ph]}}.

\bibitem{Yamazaki:2018xbx}
M.~Yamazaki, ``{Integrability As Duality: The Gauge/YBE Correspondence},'' \href{http://dx.doi.org/10.1016/j.physrep.2020.01.006}{{\em Phys. Rept.} {\bfseries 859} (2020) 1--20}, \href{http://arxiv.org/abs/1808.04374}{{\ttfamily arXiv:1808.04374 [hep-th]}}.

\bibitem{Sarkissian:2018ppc}
G.~Sarkissian and V.~P. Spiridonov, ``{From rarefied elliptic beta integral to parafermionic star-triangle relation},'' \href{http://dx.doi.org/10.1007/JHEP10(2018)097}{{\em JHEP} {\bfseries 10} (2018) 097}, \href{http://arxiv.org/abs/1809.00493}{{\ttfamily arXiv:1809.00493 [hep-th]}}.

\bibitem{Spiridonov:2019uuw}
V.~P. Spiridonov, ``{The rarefied elliptic Bailey lemma and the Yang\textendash{}Baxter equation},'' \href{http://dx.doi.org/10.1088/1751-8121/ab3358}{{\em J. Phys. A} {\bfseries 52} no.~35, (2019) 355201}, \href{http://arxiv.org/abs/1904.12046}{{\ttfamily arXiv:1904.12046 [math-ph]}}.

\bibitem{Derkachov:2019ynh}
S.~E. Derkachov and A.~N. Manashov, ``{On Complex Gamma-Function Integrals},'' \href{http://dx.doi.org/10.3842/SIGMA.2020.003}{{\em SIGMA} {\bfseries 16} (2020) 003}, \href{http://arxiv.org/abs/1908.01530}{{\ttfamily arXiv:1908.01530 [math-ph]}}.

\bibitem{Derkachov:2019tzo}
S.~Derkachov and E.~Olivucci, ``{Exactly solvable magnet of conformal spins in four dimensions},'' \href{http://dx.doi.org/10.1103/PhysRevLett.125.031603}{{\em Phys. Rev. Lett.} {\bfseries 125} no.~3, (2020) 031603}, \href{http://arxiv.org/abs/1912.07588}{{\ttfamily arXiv:1912.07588 [hep-th]}}.

\bibitem{Eren:2019ibl}
E.~Eren, I.~Gahramanov, S.~Jafarzade, and G.~Mogol, ``{Gamma function solutions to the star-triangle equation},'' \href{http://dx.doi.org/10.1016/j.nuclphysb.2020.115283}{{\em Nucl. Phys. B} {\bfseries 963} (2021) 115283}, \href{http://arxiv.org/abs/1912.12271}{{\ttfamily arXiv:1912.12271 [math-ph]}}.

\bibitem{de-la-Cruz-Moreno:2020xop}
J.~de-la Cruz-Moreno and H.~Garc\'\i{}a-Compe\'an, ``{Star-triangle type relations from $2d$ $\mathcal{N}=(0,2)$ $USp(2N)$ dualities},'' \href{http://dx.doi.org/10.1007/JHEP01(2021)023}{{\em JHEP} {\bfseries 01} (2021) 023}, \href{http://arxiv.org/abs/2008.02419}{{\ttfamily arXiv:2008.02419 [hep-th]}}.

\bibitem{Bazhanov:2022wdj}
V.~V. Bazhanov and S.~M. Sergeev, ``{An Ising-type formulation of the six-vertex model},'' \href{http://dx.doi.org/10.1016/j.nuclphysb.2022.116055}{{\em Nucl. Phys. B} {\bfseries 986} (2023) 116055}, \href{http://arxiv.org/abs/2205.10708}{{\ttfamily arXiv:2205.10708 [math-ph]}}.

\bibitem{SchlosserSTR}
M.~Schlosser, ``An elliptic extension of the multinomial theorem,'' \href{http://dx.doi.org/doi.org/10.48550/arXiv.2307.12921}{{\em arXiv:2037.12921} (2023) }.

\bibitem{Baxter:1986df}
R.~J. Baxter, ``{Free-fermion, checkerboard and Z-invariant lattice models in statistical mechanics},''
\href{http://dx.doi.org/10.1098/rspa.1986.0016}{{\em Proc. Roy. Soc. Lond.} {\bfseries A404} (1986) 1--33}.

\bibitem{Bazhanov:1992jqa}
V.~V. Bazhanov and R.~J. Baxter, ``{New solvable lattice models in three-dimensions},''
\href{http://dx.doi.org/10.1007/BF01050423}{{\em J. Statist. Phys.} {\bfseries 69} (1992) 453--585}.

\bibitem{KashaevStarSquare}
R.~Kashaev, V.~Mangazeev, and Y.~G. Stroganov, ``Star-square and tetrahedron equations in the {B}axter-{B}azhanov model,'' \href{http://dx.doi.org/10.1142/S0217751X93000588}{{\em Int. J. Mod. Phys. A} {\bfseries 08} no.~08, (1993) 1399--1409}.

\bibitem{Baxter:1997tn}
R.~J. Baxter, ``{Star-triangle and star-star relations in statistical mechanics},''
\href{http://dx.doi.org/10.1142/S0217979297000058}{{\em Int. J. Mod. Phys.} {\bfseries B11} (1997) 27--37}.

\bibitem{Bazhanov:1989nc}
V.~V. Bazhanov and Y.~G. Stroganov, ``{Chiral Potts model as a descendant of the six vertex model},'' \href{http://dx.doi.org/10.1007/BF01025851}{{\em J. Statist. Phys.} {\bfseries 59} (1990) 799--817}.

\bibitem{Bazhanov:1990qk}
V.~V. Bazhanov, R.~M. Kashaev, V.~V. Mangazeev, and Y.~G. Stroganov, ``{(Z(N)x)**(n-1) generalization of the chiral Potts model},'' \href{http://dx.doi.org/10.1007/BF02099497}{{\em Commun. Math. Phys.} {\bfseries 138} (1991) 393--408}.

\bibitem{Bazhanov:2011mz}
V.~V. Bazhanov and S.~M. Sergeev, ``{Elliptic gamma-function and multi-spin solutions of the {Y}ang-{B}axter equation},'' \href{http://dx.doi.org/10.1016/j.nuclphysb.2011.10.032}{{\em Nucl. Phys.} {\bfseries B856} (2012) 475--496},
\href{http://arxiv.org/abs/1106.5874}{{\ttfamily arXiv:1106.5874 [math-ph]}}.

\bibitem{Bazhanov:2013bh}
V.~V. Bazhanov, A.~P. Kels, and S.~M. Sergeev, ``{Comment on star-star relations in statistical mechanics and elliptic gamma-function identities},'' \href{http://dx.doi.org/10.1088/1751-8113/46/15/152001}{{\em J. Phys.} {\bfseries A46} (2013) 152001},
\href{http://arxiv.org/abs/1301.5775}{{\ttfamily arXiv:1301.5775 [math-ph]}}.

\bibitem{Kels:2017toi}
A.~P. Kels and M.~Yamazaki, ``{Elliptic hypergeometric sum/integral transformations and supersymmetric lens index},'' \href{http://dx.doi.org/10.3842/SIGMA.2018.013}{{\em SIGMA} {\bfseries 14} (2018) 013},
\href{http://arxiv.org/abs/1704.03159}{{\ttfamily arXiv:1704.03159 [math-ph]}}.

\bibitem{Catak:2021coz}
E.~Catak, I.~Gahramanov, and M.~Mullahasanoglu, ``{Hyperbolic and trigonometric hypergeometric solutions to the star-star equation},'' \href{http://dx.doi.org/10.1140/epjc/s10052-022-10661-x}{{\em Eur. Phys. J. C} {\bfseries 82} no.~9, (2022) 789}, \href{http://arxiv.org/abs/2107.06880}{{\ttfamily arXiv:2107.06880 [hep-th]}}.

\bibitem{Mullahasanoglu:2023nes}
M.~Mullahasanoglu, ``{The star\textendash{}square relation and the generalized star\textendash{}triangle relation from 3d supersymmetric dualities I},'' \href{http://dx.doi.org/10.1140/epjp/s13360-024-05444-0}{{\em Eur. Phys. J. Plus} {\bfseries 139} no.~7, (2024) 643}, \href{http://arxiv.org/abs/2306.13581}{{\ttfamily arXiv:2306.13581 [hep-th]}}.

\bibitem{Kels:2020zjn}
A.~P. Kels, ``{Interaction-round-a-face and consistency-around-a-face-centered-cube},'' \href{http://dx.doi.org/10.1063/5.0024630}{{\em J. Math. Phys.} {\bfseries 62} no.~3, (2021) 033509}, \href{http://arxiv.org/abs/2003.08883}{{\ttfamily arXiv:2003.08883 [math-ph]}}.

\bibitem{hietarinta_joshi_nijhoff_2016}
J.~Hietarinta, N.~Joshi, and F.~W. Nijhoff, \href{http://dx.doi.org/10.1017/CBO9781107337411}{{\em Discrete Systems and Integrability}}.
\newblock Cambridge Texts in Applied Mathematics. Cambridge University Press, 2016.

\bibitem{ABS}
V.~Adler, A.~Bobenko, and Y.~Suris, ``Classification of integrable equations on quad-graphs. the consistency approach,'' \href{http://dx.doi.org/10.1007/s00220-002-0762-8}{{\em Commun. Math. Phys.} {\bfseries 233} no.~3, (2003) 513--543}.

\bibitem{ABS2}
V.~E. Adler, A.~I. Bobenko, and Y.~B. Suris, ``Discrete nonlinear hyperbolic equations. classification of integrable cases,'' \href{http://dx.doi.org/10.1007/s10688-009-0002-5}{{\em Funct. Anal. Appl.} {\bfseries 43} no.~1, (2009) 3--17}.

\bibitem{AdlerPlanarGraphs}
V.~Adler, ``Discrete equations on planar graphs,'' \href{http://dx.doi.org/10.1088/0305-4470/34/48/310}{{\em J. Phys. A: Math. Gen.} {\bfseries 34} no.~48, (2001) 10453--10460}.

\bibitem{BobSurQuadGraphs}
A.~I. Bobenko and Y.~B. Suris, ``{I}ntegrable systems on quad-graphs,'' \href{http://dx.doi.org/10.1155/S1073792802110075}{{\em Int. Math. Res. Not.} {\bfseries 2002} no.~11, (2002) 573--611}.

\bibitem{MR2467378}
A.~I. Bobenko and Y.~B. Suris, {\em Discrete differential geometry: Integrable structure}, vol.~98 of {\em Graduate Studies in Mathematics}.
\newblock American Mathematical Society, Providence, RI, 2008.

\bibitem{KelsLax1}
A.~P. Kels, ``Lax matrices for lattice equations which satisfy consistency-around-a-face-centered-cube,'' \href{http://dx.doi.org/10.1088/1361-6544/ac1f76}{{\em Nonlinearity} {\bfseries 34} no.~10, (2021) 7064--7094}.

\bibitem{KelsLax2}
A.~P. Kels, ``Discrete integrable equations on face-centred cubics: consistency and {L}ax pairs of corner equations,'' {\em Proc. A.} {\bfseries 478} no.~2260, (2022) Paper No. 20210892, 22.

\bibitem{nijhoffwalker}
F.~W. Nijhoff and A.~J. Walker, ``The discrete and continuous {P}ainlev\'e {VI} hierarchy and the {G}arnier systems,'' \href{http://dx.doi.org/10.1017/S0017089501000106}{{\em Glasgow Math. J.} {\bfseries 43} no.~A, (2001) 109–--123}.

\bibitem{Kels:2018qzx}
A.~P. Kels, ``{Extended Z-invariance for integrable vector and face models and multi-component integrable quad equations},'' \href{http://dx.doi.org/10.1007/s10955-019-02346-9}{{\em J. Stat. Phys.} {\bfseries 136} (2019) 1375},
\href{http://arxiv.org/abs/1812.10893}{{\ttfamily arXiv:1812.10893 [math-ph]}}.

\bibitem{Rains2009}
E.~M. Rains, ``Limits of elliptic hypergeometric integrals,'' \href{http://dx.doi.org/10.1007/s11139-007-9055-3}{{\em The Ramanujan Journal} {\bfseries 18} no.~3, (2009) 257--306}.

\bibitem{Faddeev:1994fw}
L.~D. Faddeev, ``{Current-like variables in massive and massless integrable models},'' in {\em {Proceedings, 127th Course of the International School of Physics 'Enrico Fermi', Varenna, Italy, June 28-July 8}}, pp.~117--136.
\newblock 1994.
\newblock
\href{http://arxiv.org/abs/hep-th/9408041}{{\ttfamily arXiv:hep-th/9408041 [hep-th]}}.
\newblock

\bibitem{Ruijsenaars:1997:FOA}
S.~N.~M. Ruijsenaars, ``First order analytic difference equations and integrable quantum systems,'' \href{http://dx.doi.org/10.1063/1.531809}{{\em J. Math. Phys.} {\bfseries 38} no.~2, (1997) 1069--1146}.

\bibitem{Barnes:1901}
E.~W. Barnes, ``Theory of the double gamma function,'' {\em Phil. Trans. Roy. Soc. A} {\bfseries 196} (1901) 265--388.

\bibitem{Barnes1904}
E.~W. Barnes, ``On the theory of the multiple gamma function,'' {\em Trans. Cambridge Philos. Soc.} {\bfseries 19} (1904) 374--425.

\bibitem{shintani1976}
T.~Shintani, ``On kronecker limit formula for real quadratic fields,'' \href{http://dx.doi.org/10.3792/pja/1195518272}{{\em Proc. Japan Acad.} {\bfseries 52} no.~7, (1976) 355--358}.

\bibitem{kurokawa1991}
N.~Kurokawa, ``Multiple sine functions and selberg zeta functions,'' \href{http://dx.doi.org/10.3792/pjaa.67.61}{{\em Proc. Japan Acad. Ser. A Math. Sci.} {\bfseries 67} no.~3, (1991) 61--64}.

\bibitem{FaddeevKashaev}
L.~Faddeev and R.~Kashaev, ``Quantum dilogarithm,'' \href{http://dx.doi.org/10.1142/S0217732394000447}{{\em Mod. Phys. Lett. A} {\bfseries 09} no.~05, (1994) 427--434}.

\bibitem{RainsT}
E.~M. Rains, ``Transformations of elliptic hypergeometric integrals,'' \href{http://dx.doi.org/10.4007/annals.2010.171.169}{{\em Ann. of Math.} {\bfseries 171} (2010) 169--243}, \href{http://arxiv.org/abs/math.QA/0309252}{{\ttfamily arXiv:math.QA/0309252 [math.QA]}}.

\bibitem{GubbiottiKels}
G.~Gubbiotti and A.~P. Kels, ``Algebraic entropy for face-centered quad equations,'' \href{http://dx.doi.org/10.1088/1751-8121/ac2aeb}{{\em J. Phys. A} {\bfseries 54} no.~45, (2021) Paper No. 455201, 44}.

\bibitem{Kels:2022pdz}
A.~P. Kels, ``{Integrable systems on hexagonal lattices and consistency on polytopes with quadrilateral and hexagonal faces},'' \href{http://arxiv.org/abs/2205.02720}{{\ttfamily arXiv:2205.02720 [math-ph]}}.

\bibitem{GKV24}
G.~Gubbiotti, A.~P. Kels, and C.-M. Viallet, ``Algebraic entropy for hex systems,'' \href{http://dx.doi.org/10.1088/1361-6544/ad88cd}{{\em Nonlinearity} {\bfseries 37} no.~12, (Nov, 2024) 125007}.

\end{thebibliography}\endgroup
\bibliographystyle{utphys}
}

\end{document}